\pdfoutput=1

\RequirePackage{snapshot}
\documentclass[12pt,fabfinalm]{fabarticle}

\usepackage{amsmath,amssymb,latexsym}
\usepackage{stmaryrd,wasysym,upgreek}
\usepackage{fabmacromath}

\usepackage{graphicx}
\usepackage{subfig}
\usepackage{multirow}

\usepackage[square,numbers,sort&compress]{natbib}

\usepackage{bibentry}

\usepackage[english]{babel}

\usepackage{hyperref}
\hypersetup{%
  bookmarksopen=false,%
  bookmarksnumbered=true,%
  linktocpage=false,%
  colorlinks=true,
  linkcolor=black,%
  anchorcolor=black,%
  citecolor=black,%
  urlcolor=black,%
  pdftitle={},%
  pdfauthor={Dine Ousmane Samary, Fabien Vignes-Tourneret},%
  pdfsubject={},%
  pdfkeywords={},%
}

\numberwithin{equation}{section}

\usepackage[hyperref,thmmarks,amsmath]{ntheorem}

\usepackage{cleveref}
\crefname{thm}{theorem}{theorems}
\crefname{prop}{proposition}{propositions}
\crefname{cor}{corollary}{Corollary}
\crefname{defn}{definition}{Definition}
\crefname{defnb}{definition}{Definition}

\newcommand{\Fm}{F_{\text{max}}}
\newcommand{\Rm}{R_{\text{max}}}
\newcommand{\vp}{\varphi}
\newcommand{\vpb}{\bar{\varphi}}
\renewcommand{\rk}[1]{\textup{rank$(#1)$}}
\newcommand{\Lt}{\widetilde{L}}

\allowdisplaybreaks[1]

\title{\MakeUppercase{Just Renormalizable TGFT's on
    $U(1)^{\MakeLowercase d}$}\\\MakeUppercase{with Gauge
    Invariance}}
\author{Dine Ousmane Samary and Fabien Vignes-Tourneret}
\date{}

\begin{document}
\nobibliography*
\maketitle

\vfill
\begin{abstract}
We study the polynomial Abelian or $U(1)^d$ Tensorial Group Field Theories 
equipped with a gauge invariance condition in any 
dimension $d$.  We prove the just renormalizability at all 
orders of perturbation of the $\vp^4_6$ and $\vp^6_5$ random tensor models. We also deduce that the $\vp^4_5$ tensor model is super-renormalizable.
\end{abstract}

\vfill
\tableofcontents
\newpage
\section{Introduction}
\label{sec:introduction}

\subsection{Motivations}
\label{sec:motivations}

The complete definition of a quantum theory of gravity is probably one of the most fundamental problems of theoretical physics. 
According to several theoreticians, such a theory should obviously be background independent. As a consequence, spacetime has
to be reconstructed from more fundamental degrees of freedom which may be very well of a discrete nature.

Tensor Group Field Theory (TGFT) is quite a recent framework which aims at describing such a pre-geometric phase \citep{Rivasseau2011ab,Rivasseau2012ab}. Such an approach stands at the intersection of random tensor models and Group Field Theory (GFT). Random Tensors, especially colored 
ones, allow to define probability measures on simplicial
pseudo-manifolds (see \citep{Gurau2012ac} and references therein). 
Let us recall quickly that a random tensor of rank $d$ represents a
$(d-1)$-simplex. Each of its $d$ indices corresponds to a
$(d-2)$-simplex defining its boundary. The typical interaction part of a tensor model
is given by the gluing of $d+1$ $(d-1)$-simplices to get a 
$d$-simplex. GFT equips those tensors with
some crucial group theoretical data regarded as the seeds  
of a post geometric phase \citep{Carrozza2012aa}. TGFT could potentially relate a discrete quantum pre-geometric phase to a classical continuum limit consistent with Einstein General Relativity through a phase transition dubbed geometrogenesis.

$2D$ quantum gravity via matrix models is a successful example of such
a program. Matrix models indeed are theories of discrete surfaces yielding (after a phase transition) in the continuum, a theory of gravity dominated by sphere geometries \citep{Di-Francesco1995aa}. 
It can be stressed that the crucial analytical ingredient for achieving 
this result is the t'Hooft $1/N$ expansion. Until recently there was no analogue of such an
expansion in higher dimensions or for tensors of higher rank. Then, Gurau discovered a genuine way to generalize the matrix $1/N$ expansion to any dimension and any rank but for particular tensors \citep{Gurau2011ab,Gurau2011ac,Gurau2012aa}. Indeed, the new 1/N 
expansion relies on the so-called \emph{colored} random tensor models 
\citep{Gurau2012ac,Gurau2011aa}. 
The net result of this analysis is that the partition and correlation functions of the colored models admit perturbative expansions 
which are dominated by peculiar triangulations of spheres called \emph{melons} \cite{Bonzom2011aa}. This result has  been extended to any dimension.

Moreover, it has been realized \citep{Bonzom2012ac} that colored models
can be used to construct effective actions (and, then later \citep{Ben-Geloun2011aa}, renormalizable actions!) for
\emph{uncolored} tensor fields. In dimension $d$, there are $d+1$
colored fields. By integrating over $d$ of them, one obtains an
effective action for the last one, whose interactions are
dominated by terms corresponding precisely to those spheres which
dominate the tensor $1/N$ expansion. 

The first TGFT proved to be (just) renormalizable is a complex
$\vp^{6}$ tensor field theory on four copies of $U(1)$
\citep{Ben-Geloun2011aa}. Since then, other examples have been
discovered \citep{Ben-Geloun2012aa,Geloun2012aa,Carrozza2012aa}. In
particular, the contribution \citep{Carrozza2012aa} deals with a
propagator which implements the so-called closure or gauge
invariance condition on tensors. Such an additional symmetry is 
necessary, for instance, in order to interpret the Feynman amplitudes of the tensor model as the amplitudes of a discretized simplicial manifold issued from topological BF theories. We mention also that the model considered in \citep{Carrozza2012aa} is super-renormalizable. 
Let us shortly call these models as $\vp^n_d$, where $\vp: U(1)^d \to \mathbb{C}$ is the rank $d$ tensor and $n$ is the 
maximal coordination (or valence) of the vertices of the theory.
Our aim, in this paper, is to exhibit the first examples of \emph{just} renormalizable Abelian TGFT's on $U(1)^{d}$ with gauge invariance.

The paper is organized as follows. We first recall the basic definitions of colored graph
theory in \cref{sec:colored-graphs} and their effective (faded)
counterpart in \cref{sec:uncolored-graphs}. In \cref{sec:model}, we
present two main models analyzed in this paper, namely the
$\vp^{4}_{6}$ and $\vp^{6}_{5}$ models. \Cref{sec:power-counting} is the core of our contribution. It deals with the multi-scale analysis and the power counting theorem of some general polynomial TGFT's. Using this study, we provide the classification of divergent graphs appearing in   $\vp^{4}_{6}$ and $\vp^{6}_{5}$ which yields a control on the divergent amplitudes of these models.  
\Cref{sec:renormalization} is devoted to the renormalization of 
these divergent graphs providing, finally, the proof of  the 
renormalizability of the $\vp^{4}_{6}$ and $\vp^{6}_{5}$ models. \Cref{sec:super-ren} discusses, in a streamlined analysis, the super-renormalizability of the $\vp^4_6$ model
followed by a conclusion and two technical appendices.

\subsection{Colored graphs}
\label{sec:colored-graphs}

The Feynman graphs of the colored tensor model are $(d+1)$-colored graphs \cite{Gurau2011aa,Gurau2012ac}. For the sake
of completeness, we remind here few facts about these graphs,
their representation as stranded graphs and their uncolored version.

The graphs that we consider possibly bear external edges,
that is to say half-edges hooked to a unique vertex. We denote
$\cG_c$ a colored graph, $\cL(\cG_{c})$ 
the set of its internal edges ($L(\cG_{c})=|\cL(\cG_{c})|$) 
and $\cL_{e}(\cG_{c})$ the set of its external
edges ($L_{e}(\cG_{c})=|\cL_{e}(\cG_{c})|$). 
For all $n\in\N$, let $[n]$ be the set $\{0,1,\dots,n\}$ and 
$[n]^{*}$ be $\{1,2,\dots,n\}$.
\begin{defn}[Colored graphs]
  \label{def-ColoredGraphs}
  Let $d\in\N^{*}$. A $(d+1)$-colored graph $\cG_{c}$ is a $(d+1)$-regular bipartite graph
  equipped with a proper edge-coloring. In other words, there exists a
  map $\eta:\cL(\cG_{c})\cup\cL_{e}(\cG_{c})\to [d]$ such that if $e$ and
  $e'$ are adjacent edges, $\eta(e)\neq\eta(e')$.

A colored graph is said closed if it has no external edges and open otherwise.
\end{defn}
Examples of $4$-colored graphs are given in \cref{fig:ColoredGraphs}.
\begin{figure}[!htp]
  \centering
  \subfloat[Closed]{{\label{ClosedEx}}\includegraphics[scale=.8]{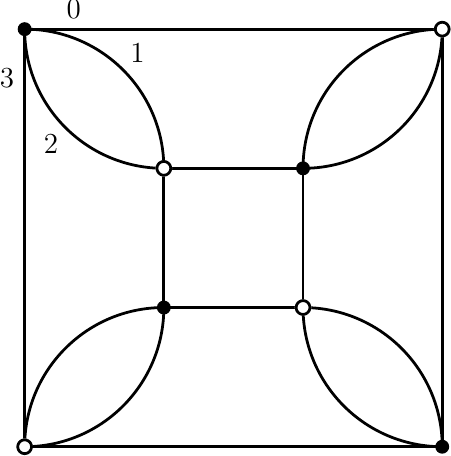}}\hspace{2cm}
  \subfloat[Open]{{\label{OpenEx}}\includegraphics[scale=1]{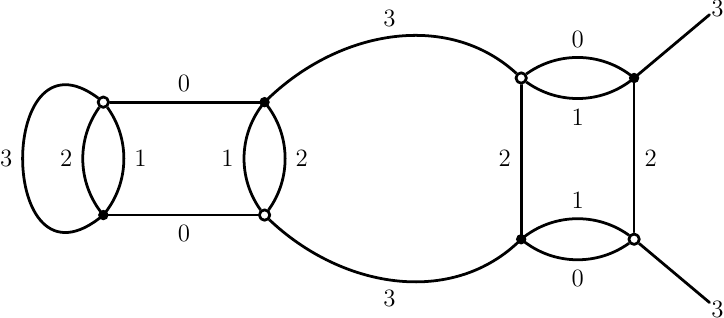}}
  \caption{Colored Graphs}
  \label{fig:ColoredGraphs}
\end{figure}

\begin{defn}[Faces]
  \label{def-faces}
  Let $\cG_{c}$ be a $(d+1)$-colored graph and $S$ a subset of
  $\{0,\dots,d\}$. We note $\cG_{c}^{S}$ the spanning subgraph of $\cG_{c}$  induced by the edges of colors in $S$.
 Then for all $0\les i,j\les  d$, $i\neq j$, a face of colors $i,j$ is a connected component of $\cG_{c}^{\{i,j\}}$.
\end{defn}
A face is open (or external) if it contains an external edge and
closed (or internal) otherwise. The set of closed faces of a graph $\cG_{c}$ is
written $\cF(\cG_{c})$ ($F(\cG_{c})=|\cF(\cG_{c})|$).
\begin{defn}[Jackets]
  \label{def-Jacket}
  Let $\sigma$ be a cyclic permutation on $[d]$, up to orientation. The jacket $J_{\sigma}$ of a
  $(d+1)$-colored graph $\cG_{c}$ is the ribbon subgraph of $\cG_{c}$ whose faces are colored $\sigma^{q}(0),\sigma^{q+1}(0)$ for $0\les q\les d$.
\end{defn}
To any jacket $J\subset\cG_{c}$, we associate its closed version
$\widetilde J$ obtained from $J$ by pinching its external legs. See \cref{fig:Jackets}.

 The numerous applications of random matrices
originate in the possibility to control (at least partially
but non perturbatively) the perturbative series of 
the partition function of these models. This interesting feature is due to the existence of the $1/N$-expansion of matrix models ($N$ denoting the size of the matrix) which provides in return a topological expansion of the partition function in terms of the genus. In higher dimensions, the generalization of such an $1/N$-expansion 
(where $N$ will denote the typical size of the tensor)
does not yield a topological expansion but rather 
a combinatorial expansion in terms of the degree
of the graph \citep{Gurau2012aa,Gurau2011ac,Gurau2011ab}.
For a colored closed graph $\cG_{c}$, the degree is defined as
\begin{align}
  \omega(\cG_{c})\defi&\sum_{J\text{ jacket of }\cG_{c}}g_{J}.\label{eq:degreeDef}
\end{align}
\begin{figure}[!htp]
  \centering
  \begin{tabular}{cc}
  \subfloat[A 2-point colored graph]{{\label{OpenEx3}}\includegraphics[scale=1.2]{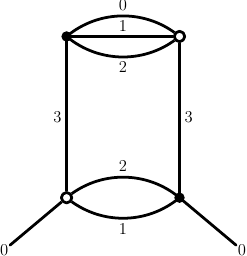}}\hspace{1cm}&
  \subfloat[Its jacket
  $(0123)$]{{\label{J1}}\includegraphics[scale=.5]{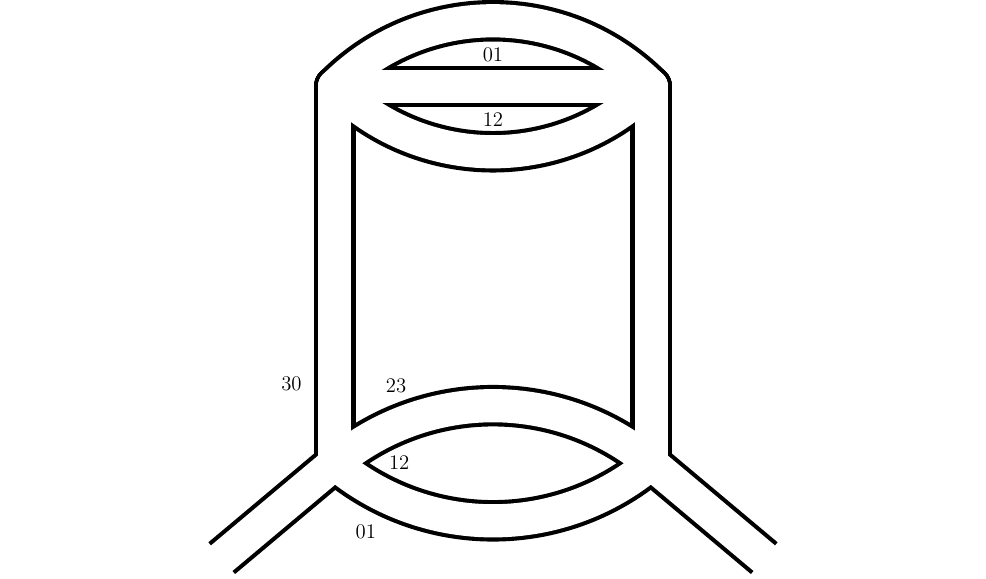}}\\
  \subfloat[Its jacket
  $(0132)$]{{\label{J2}}\includegraphics[scale=.5]{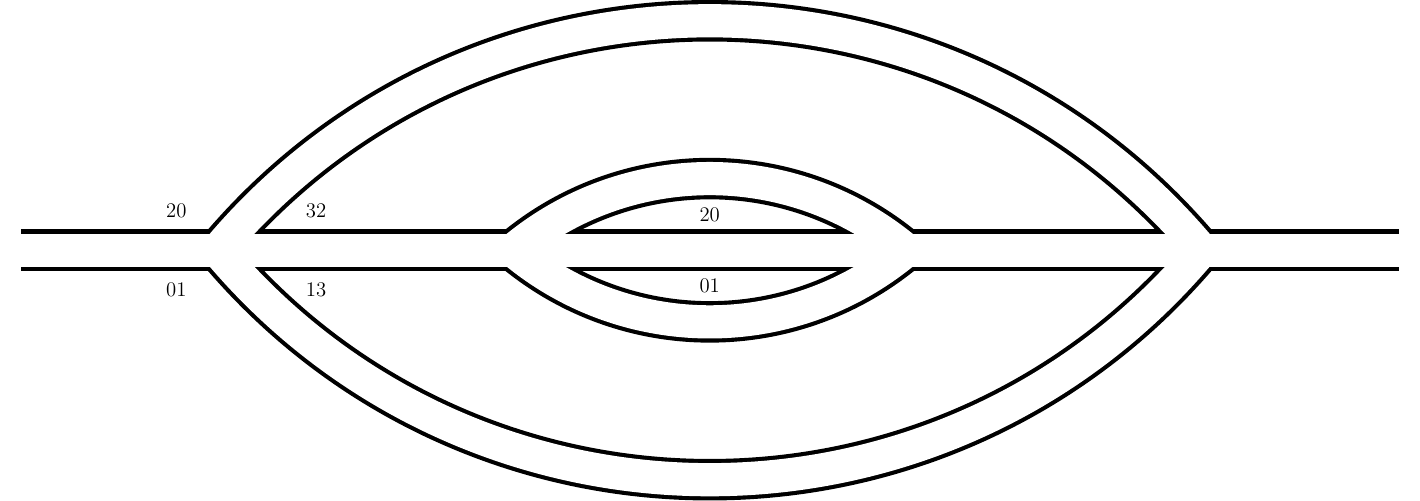}}\hspace*{1cm}&
  \subfloat[The closed jacket
  $(0123)$]{{\label{J3}}\includegraphics[scale=.5]{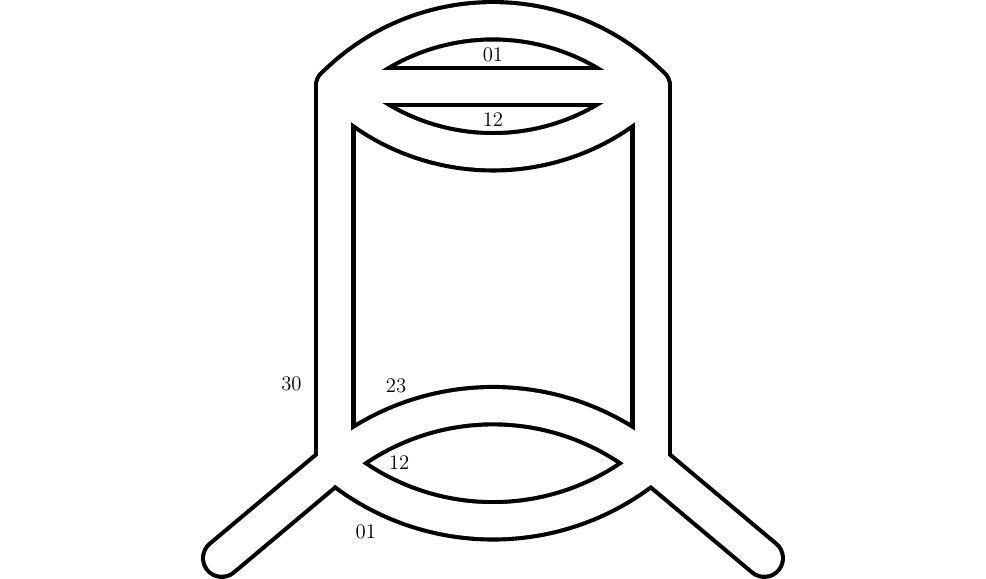}}
\end{tabular}
  \caption{Jackets}
  \label{fig:Jackets}
\end{figure}

A $(d+1)$-colored graph $\cG_{c}$ is dual to a $(d+1)$-simplicial complex corresponding to a pseudo-manifold \citep{Gurau2012ab}. The boundary of this manifold is triangulated by a complex dual to the boundary graph of $\cG_{c}$.
\begin{defn}[Boundary graph]
  \label{def-BoundaryGraph}
  Let $\cG_{c}$ be a $(d+1)$-colored graph. Its boundary graph
  $\partial\cG_{c}$ is the $d$-colored graph whose vertex-set is the set
  $\cL_{e}(\cG_{c})$ of external edges of $\cG_{c}$ and edge-set the bi-colored
  paths linking two external edges of $\cG_{c}$. In other words, an (internal)
  edge of $\partial\cG_{c}$ corresponds to an external face of $\cG_{c}$.

Note that the boundary graph of a closed colored graph is the empty graph.
\end{defn}
\begin{figure}[!htp]
  \centering
  \subfloat[An open colored graph $\cG_{c}$]{{\label{OpenEx2}}\includegraphics[scale=.9]{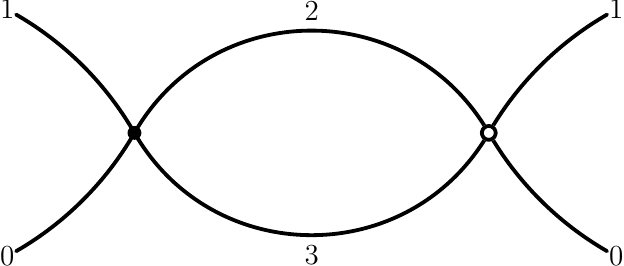}}\hspace{2cm}
  \subfloat[The boundary graph $\partial\cG_{c}$]{{\label{Boundary2}}\includegraphics[scale=.9]{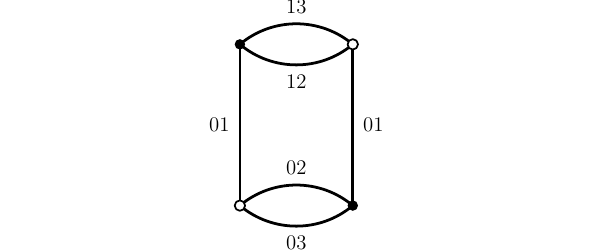}}
  \caption{The boundary operation}
  \label{fig:boundary}
\end{figure}

Any $(d+1)$-colored graph has an alternative stranded
representation. Any edge is therefore made of $d$ parallel strands. If the edge is of
color $i$, its strands are bicolored $ij$  with $j\in\hat i\defi [d]\setminus\{i\}$. The connecting pattern of any
$(d+1)$-valent vertex is the complete graph $K_{d+1}$. A (closed) face is then
represented as a (closed) curve made of one strand. An example of
such a representation is given in \cref{fig:StrandedEx}.
\begin{figure}[!htp]
  \centering
  \includegraphics[scale=.9]{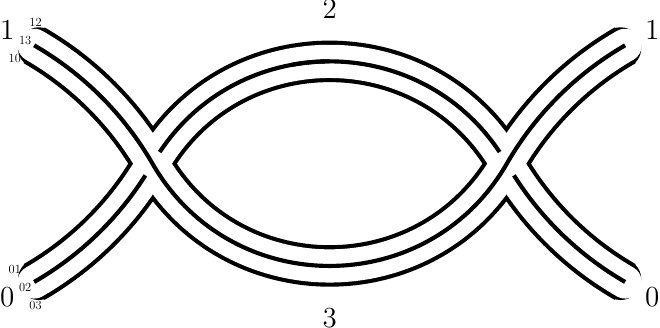}
  \caption{The stranded representation of \cref{OpenEx2}}
  \label{fig:StrandedEx}
\end{figure}

\subsection{Uncolored Graphs}
\label{sec:uncolored-graphs}

As explained in section \ref{sec:motivations}, we are interested in
effective actions obtained from the iid model \citep{Gurau2011aa} by
integrating over the fields of colors from $1$ to $d$. The effective vertices
correspond to open melonic graphs \cite{Bonzom2011aa} whose external edges are of color
$0$. The Feynman graphs of such models are so-called uncolored
graphs.  In fact, a close inspection of these uncolored graphs
show that they still possess a colored structure. Indeed, they 
are colored graphs but whose edges of colors $1,\dots,d$
are made of only one strand whereas edges of color $0$ still contain
$d$ strands. Such graphs actually represents the connecting pattern
of the indices of the tensor field of color $0$ \citep{Bonzom2012ac}. 
Generally uncolored graphs maps onto tensor trace
invariant objects \cite{Gurau2012ad}.

An uncolored graph $\cG$ has a unique colored extension $\cG_{c}$ which contains all
the faces $ij$, $0\les i,j\les d$ of a $(d+1)$-colored graph. The faces of
the uncolored graph are the $0i$-faces of its colored extension. In
\cref{UncolEx1} is depicted an uncolored open graph. The mono-stranded lines of color $i\ges 1$ are faded. Its colored extension
$\cG_c$ is shown in \cref{OpenEx3} and its (partially) stranded representation is drawn in \cref{UncolStrand1}.
\begin{figure}[!htp]
  \centering
  \subfloat[An open uncolored graph $\cG$]{{\label{UncolEx1}}\includegraphics[scale=.8]{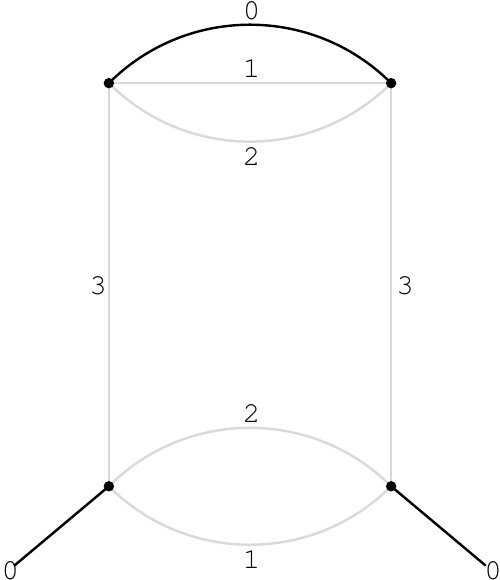}}\hspace{4cm}
  \subfloat[The stranded representation of $\cG$]{{\label{UncolStrand1}}\includegraphics[scale=.8]{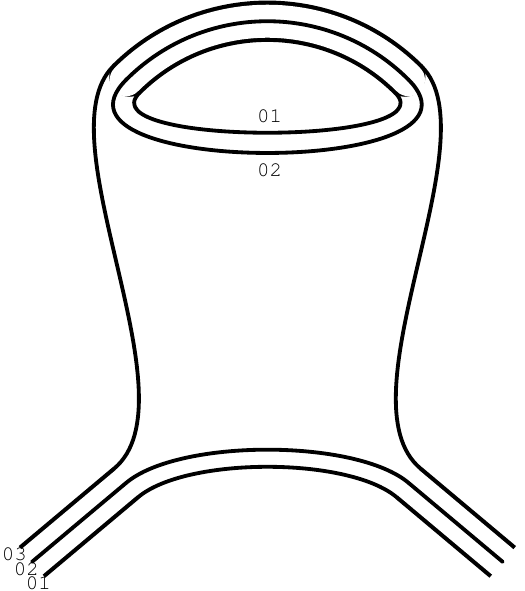}}
  \caption{An open uncolored graph $\cG$}
  \label{fig:uncolGraph}
\end{figure}

\paragraph{Connectedness}
\label{s-sect-models-let}

In graph theory, there is well-known notion of connectedness. A graph $G$
is connected if there exists at least one path in $G$ between any two
of its vertices. Another way of defining connectedness is the
following. Let us choose an orientation of the edges of $G$ and
consider the incidence matrix $I$ between edges and vertices whose
element $I_{lv}$ is $1$ if $l$ enters $v$, $-1$ if $l$ exits from $v$
and $0$ otherwise (in the case of a loop, we also choose $0$). Then a
graph is connected if it is not possible to put its incidence matrix
$I$ into a block diagonal form, after possible reordering of its rows
and columns.

There is another notion of connectedness that is relevant for tensor
graphs, colored or not. It uses the incidence matrix between edges and
faces:
\begin{defn}[Matrix $(\epsilon)_{lf}$ \citep{Carrozza2012aa}]
  \label{def-Epslf}
  Let $\cG$ be a (un)colored graph. Let us pick an arbitrary
  orientation for all of its edges and for all of its faces. 
We define the $L(\cG)\times F(\cG)$ matrix $\epsilon(\cG)$ as follows:
  \begin{align}
    \epsilon_{lf}(\cG)\defi&
    \begin{cases}
      \phantom{-}1&\text{ if $l\in f$ and their orientation match,}\\
      -1&\text{ if $l\in f$ and their orientation do not match,}\\
      \phantom{-}0&\text{otherwise.}
    \end{cases}
\label{eq:Epslf}
  \end{align}
\end{defn}
The matrix $\epsilon(\cG)$ depends on the chosen orientations but one
easily checks that its rank does not. A tensor graph will be said to
be \emph{face-connected} if its incidence matrix $\epsilon$ cannot be put
into a block diagonal form by a permutation of its rows and
columns. To distinguish between the two notions of connectedness, let
us call a graph \emph{vertex-connected} if it is connected in the
usual sense. Face-connectedness is a relevant notion in our context
because the power counting factorizes into the face-connected
components of the tensor graphs. But note that the amplitude
themselves do not enjoy such a factorization, and the usual notion of
vertex-connectedness remains relevant for renormalization (i.e. for
locality or better here traciality, see \cref{sec:trac-count}).

\newpage
\section{The models}
\label{sec:model}

Let us start now the study of quantum tensorial field theories on $U(1)^{d}$. The field in the present context is a  tensor $\vp:U(1)^{d}\to\C$. We will mainly assume that the field $\vp$ satisfies the 
following  translation invariance under a diagonal group action 
also called gauge condition:
\begin{align}\label{Gauge-condition}
\vp(hg_{1},\dots,hg_{d})=\vp(g_{1},\dots,g_{d}),\quad \forall h\in U(1)\,.
\end{align}
For $1\les j\les d$, let $g_{j}=e^{\imath\theta_{j}}\in U(1)$. 
By Fourier transform and writing $p\defi (p_{1},\dots,p_{d})$, one has
\begin{align}
  \label{eq:FourierPhi}
  \vp(g_{1},\dots,g_{d})=&\sum_{p\in\Z^{d}}\vp_{p_{1},\dots,p_{d}}\prod_{j=1}^{d}e^{\imath\theta_{j}p_{j}}.
\end{align}
A further simplification on the notation as
$\vp_{p_{1},\dots,p_{d}}\fide\vp_{1\cdots d}$
will be useful.

We will concentrate on two models namely for $\vp^{4}$ on
$U(1)^{6}$ and for $\vp^{6}$ on $U(1)^{5}$, or simply the
$\vp^4_6$ and $\vp^6_5$ models, respectively. 
We want to prove that they both are renormalizable. 
Rather than separating the renormalizability proofs, 
we perform the analysis in a row for both of these models
because of their similar features. \\

\noindent
There is a unique type of (vertex-)\emph{connected} melonic quartic
vertex in any dimension. 
In $d=6,5$, these are given by:
\begin{align}
  \label{eq:vertex}
  V^{d=6}_{4,1}\defi&\sum_{\Z^{12}}\vpb_{654321}\,\vp_{12'3'4'5'6'}\,\vpb_{6'5'4'3'2'1'}\,\vp_{1'23456}+\text{
  permutations} ,\\
\label{eq:vertex2}
 V^{d=5}_{4,1}\defi&  \sum_{\Z^{12}}\vpb_{54321}\,\vp_{12'3'4'5'}\,\vpb_{5'4'3'2'1'}\,\vp_{1'2345}+\text{
  permutations} ,
\end{align}
which are depicted in \cref{fig:Vertices4}.
 The permutations are taken on the color numbers (from $1$ to $6$ or $5$, respectively). 
\begin{figure}[!htb]
  \centering
  \includegraphics[scale=.8]{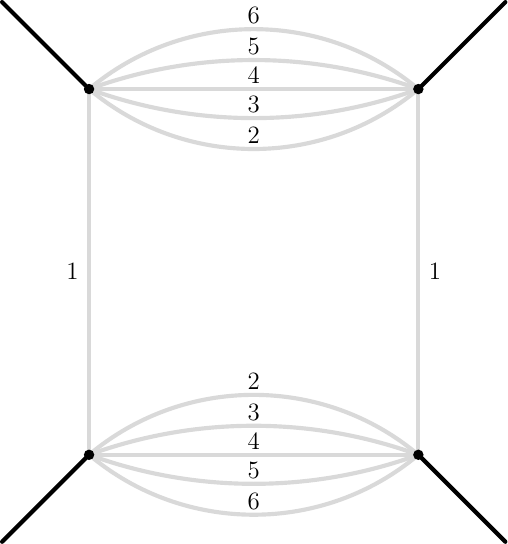} \hspace{3cm}
  \includegraphics[scale=.8]{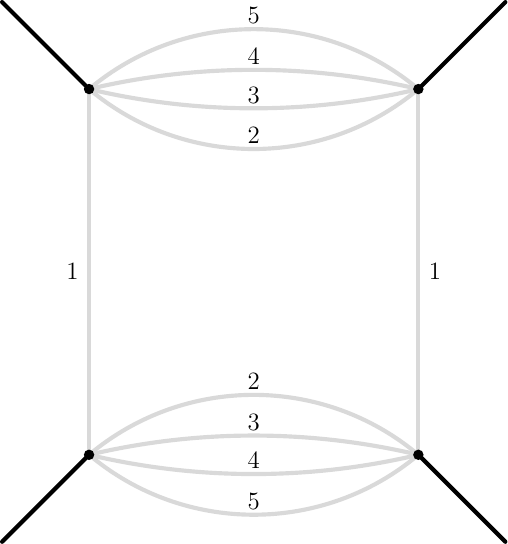}
\caption{Melonic quartic vertex for $d=6$ left and $d=5$ right}
  \label{fig:Vertices4}
\end{figure}
For the $\vp^6_5$ model, there are other interactions: two connected melonic uncolored graphs with six external edges of color $0$ (see \cref{fig:Vertices6}):
\begin{align}
  V_{6,1}\defi&\sum_{\Z^{15}}\vpb_{54321}\vp_{1'2345}\vpb_{5'4'3'2'1'}\vp_{1"2'3'4'5'}\vpb_{5"4"3"2"1"}\vp_{12"3"4"5"}+\text{
    permutations}\label{eq:V61},\\
  V_{6,2}\defi&\sum_{\Z^{15}}\vpb_{54321}\vp_{1'2345}\vpb_{5'4'3'2'1'}\vp_{1"2"3"4"5'}\vpb_{5"4"3"2"1"}\vp_{12'3'4'5"}+\text{
    permutations.}\label{eq:V62}
\end{align}
\begin{figure}[!htp]
  \centering
  \subfloat[Type I ($V_{6,1}$)]{{\label{UncolGraphs-5}}\includegraphics[scale=.8]{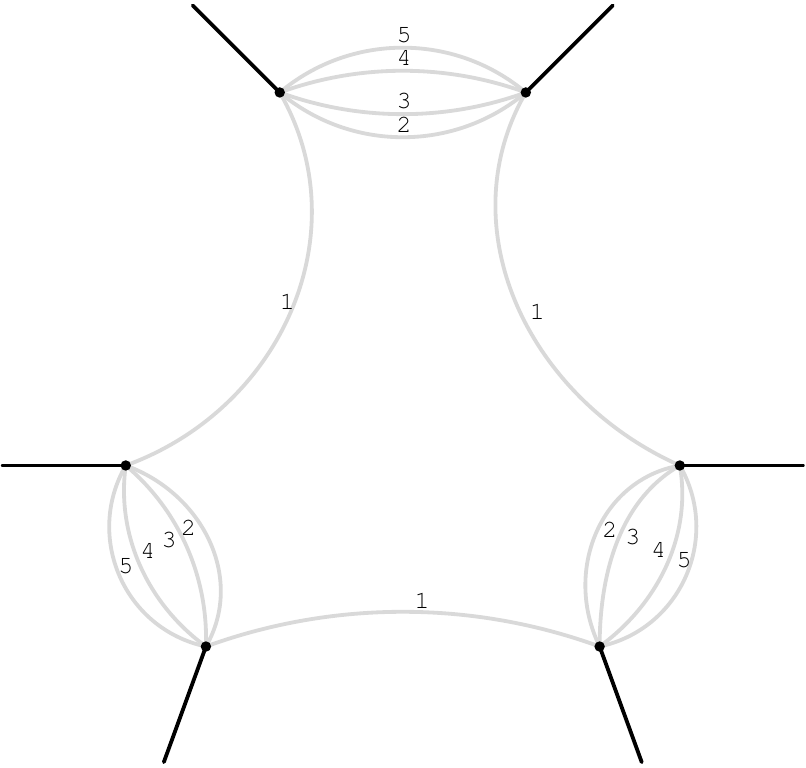}}\hspace{4cm}
  \subfloat[Type II ($V_{6,2}$)]{{\label{UncolGraphs-6}}\includegraphics[scale=.8]{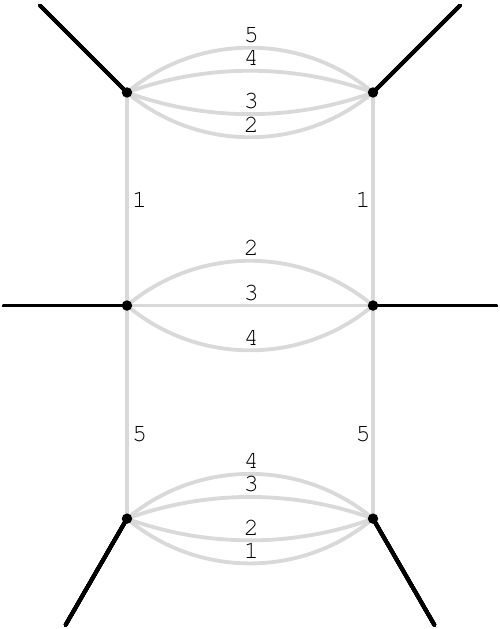}}
  \caption{Melonic Vertices of degree $6$}
\label{fig:Vertices6}
\end{figure}

As realized in \citep{Ben-Geloun2011aa}, 
the renormalization of the $4$-point function of the $\vp^{6}_{5}$ model will generate a disconnected \emph{anomalous} vertex of degree 4 (see  \cref{fig:VerticeV42})
that we need to incorporate in the action: 
\begin{align}
  \label{eq:V4anom}
  V_{4,2}\defi&\big(\sum_{\Z^{5}}\vpb_{54321}\vp_{12345}\big)^{2}.
\end{align}

\begin{figure}[!htp]
  \centering
  \includegraphics[scale=.8]{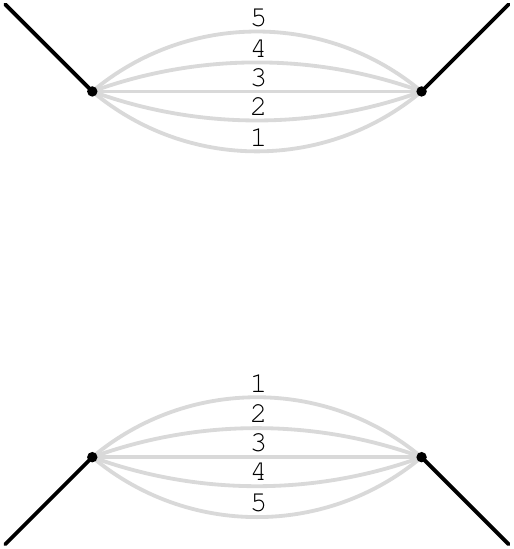}
  \caption{Graph corresponding to the vertex $V_{4,2}$  
}
  \label{fig:VerticeV42}
\end{figure}

Let $v$ be a vertex of degree $2n$ of the theory, we will
generically denote by $K_{v}$ the corresponding kernel which
is of the form
\begin{align}
  V_{2n,i}=\sum_{p_{1},\dots,p_{2nd}}K_{v}(p_{1},\dots,p_{2nd})\prod_{k=1}^{n}\vpb(p_{k,1},\dots,p_{k,d})\vp(p_{k+n,1},\dots,p_{k+n,d}).\label{eq:kernel}
\end{align}

For both models, the propagator $C(p,p')$ is the usual one $(ap^2+m^{2})^{-1}\delta(p-p')$
(where $p^2 = \sum_{i=1}^{d}p_{i}^{2}$,   
$a$ is the wave-function ``coupling constant'')
supplemented by the gauge condition 
$\delta(\sum_{i}^d p_{i})$\footnote{Note that $\delta$ here is understood as the Kronecker symbol.}.

The two actions that we consider are:
\begin{align}
  S_{4}[\vpb,\vp]=&\sum_{\Z^{6}}\vpb_{654321}\,\delta(\sum_{i}p_{i})(ap^{2}+m^{2})\,\vp_{123456}+\lambda^{(4)}_{4,1}\,V^6_{4,1},\label{eq:Action4}\\
  S_{6}[\vpb,\vp]=&\sum_{\Z^{5}}\vpb_{54321}\,\delta(\sum_{i}p_{i})(ap^{2}+m^{2})\,\vp_{12345}\nonumber\\
  &+\lambda^{(6)}_{4,1}\,V^5_{4,1}+\lambda_{4,2}V_{4,2}+\lambda_{6,1}V_{6,1}+\lambda_{6,2}V_{6,2}.\label{eq:Action6}
\end{align}

Let $\mu_{C}$ be the Gaussian measure associated to the covariance
$C$. The correlation functions are formally given by
\begin{align}\label{eq:CorrFcts}
  S_{2N}(g_{1,1},g_{1,2},\dots,g_{2N,d})=\int
  d\mu_{C}\,\big(\prod_{i=1}^{N}\vpb(g_{i,1},\dots,g_{i,d})\vp(g_{2i,1},\dots,g_{2i,d})\big)
  e^{-S_{\text{int}}[\vpb,\vp]}.
\end{align}
where $S_{\text{int}}[\vpb,\vp]$ is either $\lambda^{(4)}_{4,1}V^6_{4,1}$ or
$\lambda^{(6)}_{4,1}V^5_{4,1}+\lambda_{4,2}V_{4,2}+\lambda_{6,1}V_{6,1}+\lambda_{6,2}V_{6,2}$,
depending on the model under consideration. Our aim here is to define these correlation
functions as formal power series. In other words, we prove that the
models \eqref{eq:Action4}  and \eqref{eq:Action6} are renormalizable to all orders of
perturbation:
\begin{thm}
  \label{thm-PertRen}
  There exist formal power series $F_{1}, F_{2}, F_{3}$ in a
  parameter $\lambda_{4,1}^{(4),r}$ and multi-power series
  $\{G_{i}\}_{1\les i\les 6}$ in parameters
  $\lambda^{r}\defi\{\lambda_{4,1}^{(6),r},\lambda_{4,2}^{r},\lambda_{6,1}^{r},\lambda_{6,2}^{r}\}$
  such that if we fix
    \begin{align}
      \lambda^{(4)}_{4,1}=&F_{1}(\lambda_{4,1}^{(4),r}),\quad
      m^{2}=F_{2}(\lambda_{4,1}^{(4),r}),\quad
      a=F_{3}(\lambda_{4,1}^{(4),r}),\\
      \lambda_{4,1}=&G_{1}(\lambda^{r}),\quad
      \lambda_{4,2}=G_{2}(\lambda^{r}),\quad
      \lambda_{6,1}=G_{3}(\lambda^{r}),\\
      \lambda_{6,2}=&G_{4}(\lambda^{r}),\quad m^{2}=G_{5}(\lambda^{r}),\quad
      a=G_{6}(\lambda^{r}),
    \end{align}
all correlation functions are well-defined
  formal power series in $\lambda_{4,1}^{(4),r}$, and 
$\lambda^{r}$, respectively.
\end{thm}
The rest of the paper is devoted to the proof of this theorem.
For simplicity  and when no confusion may occur, 
both $\lambda^{(4)}_{4,1}$ and $\lambda^{(6)}_{4,1}$
will be simply denoted by $\lambda_{4,1}$.

\section{Multi-scale analysis and power counting theorem}
\label{sec:power-counting}

The goal of this section is the classification of all the primitively divergent graphs generated by both models
\eqref{eq:Action4} and \eqref{eq:Action6}. Our main tool 
is the multiscale analysis. This will help to the proof of an upper bound 
on the amplitude of a general graph implying
the existence of a power counting theorem. 

All the framework of  \cref{sec:model} directly extends to 
models with arbitrary rank tensors with polynomial interactions $P(\vpb,\vp)$. We perform our multi-scale analysis in this general 
setting and only at the end we will specialize the rank and 
maximal degree of the vertices. 
This leads us to  some 
models of interest ($\vp^4_6$, $\vp^6_5$ and $\vp^4_5$ models).

\subsection{Multiscale analysis}
\label{sec:multiscale-analysis}

Multiscale analysis  allows us to study precisely the amplitudes of
Feynman graphs through the glass of a discrete version of Hepp
sectors \citep{Riv1}. To this aim, the first step consists in slicing the
propagator into different scales. 

Let $M\in\R,\,M>1$, we have:
\begin{align}
  C(p)=&\delta\big(\sum_{i=1}^{d}p_{i}\big)\int_{\R_{+}}e^{-\alpha(p^{2}+m^{2})}d\alpha \fide\sum_{j=0}^{\infty}C^{j}(p)\label{eq:prop}
\end{align}
with
\begin{align}
 C^{0}(p)\defi&\delta\big(\sum_{i=1}^{d}p_{i}\big)\int_{1}^{\infty}e^{-\alpha(p^{2}+m^{2})}d\alpha,\\
\forall j\ges 1, \quad C^{j}(p)\defi&\delta\big(\sum_{i=1}^{d}p_{i}\big)\int^{M^{-2(j-1)}}_{M^{-2j}}e^{-\alpha(p^{2}+m^{2})}d\alpha.
\end{align}
We regularize the models with an ultraviolet cutoff by restricting the
sum over $j$ to the range $0$ to $\rho<\infty$. The momenta are thus
(smoothly) bounded by $M^{2\rho}$.
The sliced propagator admits a simple upper bound:
\begin{align}
  C^{i}(p)\les  K\,M^{-2i}e^{-M^{-2i}(p^{2}+m^{2})}\delta\big(\sum_{j}p_{j}\big),
\label{eq:propaBound}
\end{align}
where $K = M^{2}-1$.

The next stage is to bound any graph amplitude.
Consider then $\cG$ a Feynman graph. Its amplitude writes
\begin{align}
  A_{\cG}=\sum_{p_{1},\dotsc, p_{dL(\cG)}}\;\prod_{l\in\cL(\cG)}C_{l}(p_{l},p'_{l})\prod_{v\in\cV(\cG)}K_{v}\label{eq:AG},
\end{align}
where $\cV(\cG)$ ($V(\cG)=|\cV(\cG)|$) denotes the set of vertices of $\cG$ and $p_{1},\dotsc,p_{dL(\cG)}$  the momenta associated to the strands of lines in $\cG$.
 As
each propagator is sliced according to \eqref{eq:prop}, the amplitude
can be decomposed as a sum over the so-called momentum  attributions:
\begin{align}
  A_{\cG}=\sum_{i_{1},\dotsc,i_{L}}\sum_{p_{1},\dotsc,p_{dL(\cG)}}\prod_{l\in\cL(\cG)}C^{i_{l}}_{l}(p_{l},p'_{l})\prod_{v\in\cV(\cG)}K_{v}\fide\sum_{\mu\in\N^{L}}A_{\cG}^{\mu}.
\end{align}
We focus on $A_{\cG}^\mu$. The significant upper bound of the following will be expressed in
terms of certain special subgraphs of $\cG$ called 
\emph{dangerous} subgraphs defined as follows. 
Let $\cG^{\mu}$ be a Feynman graph
with a scale attribution $\mu$. For all $i\in\N$, let $\cG^{i}$ be the
subgraph of $\cG$ induced by $\cL^{i}(\cG)\defi\{l\in\cL(\cG)\tqs
i_{l}\ges i\}$. $\cG^{i}$ may have several (say of number $C_{i}(\cG)$) vertex-connected components in which case we note
them $\cG^{i}_{k}$, $1\les k\les C_{i}(\cG)$. These connected subgraphs are
the dangerous subgraphs in the sense that the power counting will be
written only in terms of those subgraphs and no other.
There is a simple way to determine if a given subgraph $\cH$ is dangerous or
not. Let $i_{\cH}(\mu)\defi\inf_{l\in\cL(\cH)}i_{l}$ and
$e_{\cH}(\mu)\defi\sup_{l\in\cL_{e}(\cH)}i_{l}$. $\cH$ is dangerous if
and only if $i_{\cH}(\mu)>e_{\cH}(\mu)$.

The $\cG^{i}_{k}$'s are partially ordered by inclusion and form in
fact a forest, i.e. a set of connected graphs such that any two of
them are either disjoint or included one in the other \citep{Riv1}. If $\cG$ is
itself connected, the forest is in fact a tree whose root is the full
graph $\cG=\cG^{0}$. This abstract tree is named the Gallavotti-Nicol\`o
(GN) tree.

Our goal is to find an optimal (with respect to a scale attribution) upper bound on the amplitude of a
general  graph $\cG^{\mu}$.
\begin{thm}[Power counting]
  \label{thm-DivDegree}
  Let $\cG$ be a Feynman graph of a polynomial $P(\vpb,\vp)$ model
  with propagator \eqref{eq:prop} on $U(1)^{d}$. There exist constants $K_{1},K_{2},K_{3}\in\R_{+}^{*}$ such that
  \begin{align}
    A_{\cG}^{\mu}\les&K_{1}^{V(\cG)}K_{2}^{N(\cG)}K_{3}^{F(\cG)}\prod_{i=0}^{\rho}\prod_{k=1}^{C_{i}(\cG)}M^{\omega_{d}(\cG^{i}_{k})},\text{
    where }\omega_{d}(\cG^{i}_{k})=-2L(\cG^{i}_{k})+F(\cG^{i}_{k})-R^{i}_{k}\label{eq:DivDegree}
  \end{align}
and $R^{i}_{k}$ is the rank of $\epsilon(\cG^{i}_{k})$.
\end{thm}
The proof of this theorem is already available in the literature. 
In \citep{Ben-Geloun2010aa}, the \emph{superficial} degree of divergence 
 of a TGFT graph amplitude
for an Abelian theory and without $(p^2+m^2)$ term 
is computed and proven to be $F-R$, with $R=\rk{\epsilon(\cG)}$. In \citep{Ben-Geloun2011aa} an
\emph{optimized} bound on the amplitudes of a $\vp^{6}$-type model
with $(p^{2}+m^{2})^{-1}$ as propagator (but without gauge invariant
condition) is proven. The degree of
divergence is there $-2L+F$. More recently, Abelian 
theories with both $(p^{2}+m^{2})^{-1}$ and gauge invariant
condition has been finally proven in \citep{Carrozza2012aa}. It is precisely the bound
\eqref{eq:DivDegree}. However, we think that it may be instructive to collect here all the arguments and rewrite the complete proof, in momentum space. 

\begin{proof}[of  \cref{thm-DivDegree}]
  We want to bound the amplitude of graph $\cG$ with scale attribution
  $\mu$:
  \begin{align}
  A^{\mu}_{\cG}=\sum_{p_{1},\dotsc,p_{dL(\cG)}}\prod_{l\in\cL(\cG)}C^{i_{l}}_{l}(p_{l},p'_{l})\prod_{v\in\cV(\cG)}K_{v}.
\end{align}
The specific forms of the vertex kernels $K_{v}$ considered in such 
models imply that there is actually one independent sum per closed
face of $\cG$ (as it is the case in matrix models). Let us pick an
arbitrary orientation of the faces and define the unique momentum of
the face $f$ to be $p_{f}$ in the direction of the chosen
orientation. The orientations (signs) of the line  momenta are similarly fixed by
choosing an orientation of the edges of $\cG$. For each line
$l\in\cL$, the delta function $\delta_{l}\big(\sum_{i=1}^{d}p_{l,i}\big)$
can be rewritten as
$\delta_{ l}\big(\sum_{f\in\cF}\epsilon_{lf}p_{f}+p_{l,e}\big)$ where $p_{l,e}$ is the sum of momenta of the line $l$ which belong to
external faces.

Using the bound \eqref{eq:propaBound} on the sliced propagator, we get
\begin{align}
  A^{\mu}_{\cG}\les&K^{L}\sum_{p_{f_{1}},\dots,p_{f_{F}}}\prod_{l\in
    \cL(\cG)}\big[M^{-2i_{l}}e^{-M^{-2i_{l}}p_{l}^{2}}\delta_{l}\big(\sum_{f\in\cF(\cG)}
\epsilon_{lf}p_{f}+p_{l,e}\big)\big]
\label{eq:AmplitudeFaceVariables}\\
  \les&K_{1}^{V}K_{2}^{N}\sum_{p_{f_{1}},\dots,p_{f_{F}}}\big[\big(\prod_{f\in\cF(\cG)}e^{-M^{-2i_{f}}p_{f}^{2}}\big)\prod_{l\in\cL}M^{-2i_{l}}
\delta_{l}\big(\sum_{f\in\cF(\cG)}\epsilon_{lf}p_{f}
+p_{l,e}\big)\big],
\end{align}
where $i_{f}\defi\inf_{l\in f}i_{l}$, $L = L(\cG)$, $V = V(\cG)$ 
and $N=N(\cG)$.  We choose subsets $\cF_{\mu}\subseteq\cF$ and $\cL_{\mu}\subseteq\cL$
such that $|\cF\setminus\cF_{\mu}|+|\cL_{\mu}|=F(\cG)$. If
$f\in\cF\setminus\cF_{\mu}$, the sum over $p_{f}$ is performed using
the corresponding exponential function. If not, the sum over $p_{f}$
is performed using a $\delta_{l}$ function corresponding to a line $l\in\cL_{\mu}$:
\begin{align}
  A^{\mu}_{\cG}\les&K_{1}^{V}K_{2}^{N}\prod_{l\in\cL}M^{-2i_{l}}\sum_{p_{f_{1}},\dots,p_{f_{F}}}\prod_{f\in\cF\setminus\cF_{\mu}}e^{-M^{-2i_{f}}p_{f}^{2}}\prod_{l\in\cL_{\mu}}\delta_{l}\big(\sum_{f\in\cF(\cG)}\epsilon_{lf}p_{f}+p_{l,e}\big).\label{eq:AmplitudeR}
\end{align}
The maximal number of sums we can perform with the $\delta_{l}$
functions is precisely $\rk{\epsilon(\cG)}$. A sum performed
with an exponential function brings a factor $M^{i_{f}}$ whereas a
sum performed with a delta function gives $1$. It is thus
necessary to optimize the choice of the sets $\cF_{\mu}$ and
$\cL_{\mu}$ with respect to the scale attribution $\mu$. In
\citep{Carrozza2012aa} it is proven that such an optimal choice is
possible and given by:
\begin{enumerate}
\item There exists a subset $\cL_{\mu}\subset\cL$ with
  $|\cL_{\mu}|=\rk{\epsilon(\cG)}$ and the arguments of the
  corresponding $\delta_{l\in\cL_{\mu}}$ functions are independent.
\item For all $i,k$, $|\cL_{\mu}\cap\cL(\cG^{i}_{k})|=\rk{\epsilon(\cG^{i}_{k})}$.
\end{enumerate}
Let us rephrase the proof of \citeauthor{Carrozza2012aa} in the
following way. For all $\cF'\subset\cF(\cG)$, let us denote
by $\epsilon_{|\cF'}$ the matrix $\epsilon(\cG)$ with columns
restricted to faces in $\cF'$. We first choose $\cF_{\mu}$, inductively
from the leaves of the GN  tree towards its root. 
Consider a leaf of the GN tree. It corresponds to a certain
$\cG^{i}_{k}$. We choose  $\rk{\epsilon(\cG^{i}_{k})}$ independent
columns of $\epsilon(\cG^{i}_{k})$. The corresponding faces of $\cG$
are put in $\cF_{\mu}$. Note that these columns are also independent
in $\epsilon_{|\cF(\cG^{i}_{k})}$. This later matrix contains indeed only
zeros on the lines $l\notin\cL(\cG^{i}_{k})$. We proceed similarly for
all the leaves of the GN tree. Then, when several $\cG^{i}_{k}$'s
merge into a $\cG^{j}_{k'}, j<i$, we add to $\cF_{\mu}$ as many faces
as necessary to have
$|\cF_{\mu}\cap\cF(\cG^{j}_{k'})|=\rk{\epsilon(\cG^{j}_{k'})}$. At the
last step of this process, when one reaches the root of the GN tree,
the cardinal of $\cF_{\mu}$ is clearly equal to the rank of
$\epsilon(\cG)$. Moreover for all $i,k$,
$|\cF_{\mu}\cap\cF(\cG^{i}_{k})|=\rk{\epsilon(\cG^{i}_{k})}$.

It remains to choose the set $\cL_{\mu}$. The matrix
$\epsilon(\cG)_{|\cF_{\mu}}$ has the same rank as
$\epsilon(\cG)$. There exist $|\cF_{\mu}|$ lines of
$\epsilon(\cG)_{|\cF_{\mu}}$ such that the restricted square matrix
has still the rank of $\epsilon(\cG)$. These lines form the set
$\cL_{\mu}$. The point is that it is possible to choose these
$|\cF_{\mu}|$ lines such that
$|\cL_{\mu}\cap\cL(\cG^{i}_{k})|=\rk{\epsilon(\cG^{i}_{k})}$ for all
$i$ and $k$. Indeed, if there exists a $\cG^{i}_{k}$ such that
$|\cL_{\mu}\cap\cL(\cG^{i}_{k})|<\rk{\epsilon(\cG^{i}_{k})}$ then
there is a line $l\in\cL(\cG^{i}_{k})\setminus\cL_{\mu}$ such that the
corresponding line-vector is independent of the $|\cL_{\mu}|$ other
ones (remember that $\epsilon(\cG)_{lf}=0$ if
$l\notin\cL(\cG^{i}_{k})$ and $f\in\cF(\cG^{i}_{k})$). And the set
$\cL_{\mu}$ of line-vectors is not maximally independent.

The proof of theorem \eqref{thm-DivDegree} is achieved by
the following. Start
from equation \eqref{eq:AmplitudeR} and write
\begin{align}
    A^{\mu}_{\cG}\les&K_{1}^{V}K_{2}^{N}K_{3}^{F}\prod_{l\in\cL}M^{-2i_{l}}\prod_{f\in\cF\setminus\cF_{\mu}}M^{i_{f}}\\
    \les&K_{1}^{V}K_{2}^{N}K_{3}^{F}\prod_{i,k}\prod_{l\in\cL(\cG^{i}_{k})}M^{-2}\prod_{i,k}\prod_{f\in\cF(\cG^{i}_{k})\setminus\cF_{\mu}}M\\
    \les&K_{1}^{V}K_{2}^{N}K_{3}^{F}\prod_{i,k}M^{-2L(\cG^{i}_{k})+F(\cG^{i}_{k})-R^{i}_{k}}.
\end{align}
\end{proof}

\subsection{Analysis of the divergence degree}
\label{sec:analys-diverg-degr}

The  divergence degree is $\omega_{d}=-2L+F-R$. In this section, 
we scrutinize this quantity and re-express it in term 
of more useful quantities. We develop as well new tools
for this task. This allows us to go beyond
the analysis in dimension $d=4$ as performed in \citep{Carrozza2012aa} 
and find renormalizable theories. 

\begin{lemma}[Contraction of a tree]
  \label{thm-TreeContraction}
  Let $\cG$ be a connected uncolored graph and $\cT$ be any of its
  spanning trees. Under contraction of $\cT$, neither $F$ nor $R$
  changes:
  \begin{align}
    \label{eq:ContractionTree}
    F(\cG)=&F(\cG/\cT),&&R(\cG)=\rk{\epsilon(\cG)}=\rk{\epsilon(\cG/\cT)}=R(\cG/\cT).
  \end{align}
\end{lemma}
\begin{proof}
 The fact that $F(\cG)$ does not change under contraction is quite obvious: under contraction of an internal line, faces can only get
  shorter. This is true both for open and closed faces.
 Moreover, if the contracted line is a tree line, the face
cannot disappear. 

  Let $\ell\in\cL(\cG)$ be any line of $\cG$ (not necessarily a tree line). The matrix $\epsilon(\cG/\ell)$
  is obtained from $\epsilon(\cG)$ by erasing the row $\ell$ and the
  columns full of zeros corresponding to the faces which disappeared
  under the contraction. In the case of a tree line, this second step
  does not happen, as explained just above. As a consequence to prove
  that $R$ is invariant under the contraction of a tree line $l$, we need
  to prove that erasing this row does not change the rank of
  $\epsilon$ that is to say that the row $l$ is a linear combination
  of the other rows of the matrix:
  \begin{align}
    \forall f\in\cF(\cG),\,\quad \epsilon_{lf}=\sum_{\substack{\ell\in\cL,\\\ell\neq
      l}}a_{l\ell}\,\epsilon_{\ell f}\,,
\label{eq:treeCL1}
  \end{align}
where the $a_{l\ell}$'s are independent of $f$.

Any oriented line $\ell$ links a vertex $v_{\ell}$ to another
(different) one $v'_{\ell}$. There is a unique oriented path $\cP_{\cT}(\ell)$ in $\cT$ from $v_{\ell}$ to
$v'_{\ell}$ (see appendix \ref{sec:paths-graph}). Thus $\cP_{\cT}$ is, in particular, a map from $\cL(\cG)$ to
$2^{\cL(\cT)}$. For any internal face $f$, the set of lines of $\cG$
contributing to this face forms a cycle. This cycle can be
\emph{projected} onto a path in $\cT$ thanks to the map
$\cP_{\cT}$. The face $f$ being a cycle, the corresponding path in
$\cT$ begins and ends at the same vertex. But as $\cT$ is acyclic,
each edge has to be covered an even number of times and in
\emph{opposite} directions. Thus if we go all over an internal face,
and count with signs the number of times a given tree line appears in
the projected path, we find zero. 

Let us pick up a face $f$ and go all over it according to its
orientation\footnote{Remember that an orientation has been chosen for
  each face and each line of $\cG$ in order to define the matrix
  $\epsilon$. We will refer to this choice as an orientation in $\cG$.}. For all $\ell\in f$ and all $l\in\cT$, let
$\veps_{\ell l}(f)$ be $+1$ if $l\in\cP_{\cT}(\ell)$ and its orientation
in $\cP_{\cT}(\ell)$ matches its chosen orientation in $\cG$, $-1$ if
$l\in\cP_{\cT}(\ell)$ and the two orientations do not match, and $0$
otherwise. We have
\begin{align}
  \sum_{\ell\in f}\veps_{\ell l}(f)=0.\label{eq:CountedWithSigns}
\end{align}
For all $\ell\in\cL$ and $l\in\cT$, let us define $\eta_{\ell l}$ as
$+1$ if $l\in\cP_{\cT}(\ell)$ and the orientation of $l$ in
$\cP_{\cT}(\ell)$ (fixed by the chosen orientation of $\ell$ in $\cG$)
matches its orientation in $\cG$, $-1$ if $l\in\cP_{\cT}(\ell)$ and
the orientations do not match, $0$ otherwise. It is not difficult to
check that $\veps_{\ell l}(f)=\eta_{\ell l}\epsilon_{\ell f}$. As
$\eta_{ll}=1$, we get
\begin{align}
  \sum_{\ell\in f}\veps_{\ell l}(f)=\sum_{\ell\in\cL}\eta_{\ell
    l}\epsilon_{\ell f}=0 \quad \Longleftrightarrow \quad \epsilon_{lf}=-\sum_{\ell\in\cL,\,\ell\neq l}\eta_{\ell
    l}\epsilon_{\ell f}\label{eq:CL}
\end{align}
which is of the form of \cref{eq:treeCL1} and achieves the proof. 
(For an example treated in detail, see appendix \ref{sec:paths-graph}.)
\end{proof}
\begin{defn}[$k$-dipole]
  \label{def-kdipole}
  Let $\cG$ be an uncolored graph. A $k$-dipole is a line $\ell$ of $\cG$
  such that it belongs to exactly $k$ faces of length $1$. In other
  words, if $\ell$ joins to vertices $v$ and $v'$ of the colored
  extension $\cG_{c}$, there are exactly $k$ edges in $\cG_{c}$ of colors $i>0$ linking $v$ and $v'$, see \cref{fig:kdipole}.
\end{defn}
\begin{figure}[!htp]
  \centering
  \includegraphics[scale=1.2]{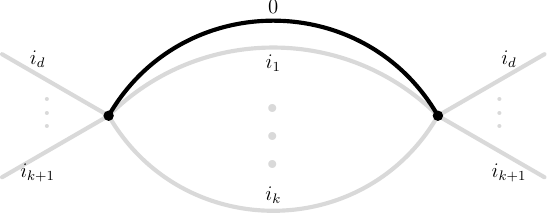}
  \caption{A $k$-dipole}
  \label{fig:kdipole}
\end{figure}

\begin{defn}[Rosettes \citep{Ben-Geloun2012ab}]
\label{def-Rosettes}
Let $\cG$ be a connected graph and $\cT$ any of its spanning trees. The
contracted graph $\cG/\cT$ is called a rosette. A rosette with external lines\footnote{For the vacuum rosette the definition is the same except that the last line $i=L-V+1$ corresponds to a $d$-dipole not a $d-1$.} is
\emph{\textbf{fully melonic}} if there exists an order on its $L-V+1$ lines
such that $l_{1}$ is a $(d-1)$-dipole in $\cG/\cT$ and for all $2\les i\les L-V+1$, $l_{i}$ is a $(d-1)$-dipole
in $\cG/(\cT\cup\{l_{1},\dots,l_{i-1}\})$.
\end{defn}

Let us consider a polynomial $P(\vpb,\vp)_{U(1)^{d}}$ model and $\cG$ one of its
graphs. For all $i\ges 2$, we denote by $V_{i}$ its number of vertices
of degree $i$ and $n\cdotaction V\defi\sum_{i\ges 2}iV_{i}$.
The following statement holds.
\begin{lemma}\label{thm-FullyMelonicDivDegree}  
  Let $\cG$ be a connected Feynman graph and $\cT$ one of its spanning
  trees. If the rosette $\cG/\cT$ is fully melonic,
  \begin{subequations}
    \begin{align}
      F(\cG)=&(d-1)(L-V+1),\\
      R(\cG)=&\Rm(\cG)\defi L-V+1,\\
      2\omega_{d}(\cG)=&-(d-4)N+(d-4)n\cdotaction
      V-2(d-2)V+2(d-2).\label{eq:FullymelonicDegree}
    \end{align}
  \end{subequations}
\end{lemma}
\begin{proof}
  We contract successively all the lines of $\cG/\cT$. The rosette
  being fully melonic, we contract only $(d-1)$-dipoles. Then for all
  $i\in[L-V+1]^{*}$,
  \begin{align}
    F(\cG/(\cT\cup\{l_{1},\dots,l_{i}\}))=&F(\cG/(\cT\cup\{l_{1},\dots,l_{i-1}\}))-(d-1),\\
    R(\cG/(\cT\cup\{l_{1},\dots,l_{i}\}))=&R(\cG/(\cT\cup\{l_{1},\dots,l_{i-1}\}))-1,\\
    F(\cG)=&F(\cG/\cT)=(d-1)(L-V+1),\label{eq:Ffullymelonics}\\
    R(\cG)=&R(\cG/\cT)=L-V+1.
  \end{align}
Using $\omega_{d}=-2L+F-R$ and $2L+N=n\cdotaction V$, one gets the desired result.
\end{proof}

We are in position to understand why the $\vp^4_6$ 
and $\vp^6_5$ are just renormalizable. 
Indeed, applying formula \eqref{eq:FullymelonicDegree} to the models
\eqref{eq:Action4} and \eqref{eq:Action6}, we get
\begin{align}
  \omega_{d,4}=&-(N-4),\quad\omega_{d,6}=-\frac{N-6}{2}-V_{4},\label{eq:FullyMelonicPowModels}
\end{align}
which are typical divergence degrees of just renormalizable models. In
the following, we will prove that the divergence degree of a graph is bounded
from above by the divergence degree of the graphs with fully melonic
rosettes. Moreover we will see that the model \eqref{eq:Action6} contains
subdivergent contributions i.e.\@ divergent graphs with non fully
melonic rosettes.\\

Let $\rho(\cG)$ be defined as $F(\cG)-R(\cG)-(d-2)\Lt(\cG)$ with
$\Lt(\cG)\defi L(\cG)-V(\cG)+1$. Note that for any spanning tree $\cT$
in $\cG$, $\Lt(\cG)=L(\cG/\cT)=\Lt(\cG/\cT)$ so that thanks to
\cref{thm-TreeContraction},
$\rho(\cG)=\rho(\cG/\cT),\,\forall\cT$. If $\cG$ is
face-disconnected, $\cG=\bigcup_{i\in I}\cG^{i}$, then
$\rho(\cG)=\sum_{i\in I}\rho(\cG^{i})$. Moreover \citeauthor*{Carrozza2012ab} have proven the following \citep{Carrozza2012ab}
\begin{lemma}
  \label{thm-rho}
  Let $\cG$ be a face-connected rosette.
  \begin{enumerate}
  \item If $N(\cG)=0$, then $\rho(\cG)\les 1$ and $\rho(\cG)=1$
    iff $\cG$ is fully melonic.
  \item If $N(\cG)>0$, then $\rho(\cG)\les 0$ and $\rho(\cG)=0$
    iff $\cG$ is fully melonic.
  \end{enumerate}
\end{lemma}
 The
divergence degree of a graph rewrites as
\begin{align}
  \omega_{d}(\cG)=&-2L(\cG)+(d-2)\Lt(\cG)+\rho(\cG)
\end{align}
which leads to
\begin{align}
  \omega_{d,4}(\cG)=4-N+\rho(\cG),\quad \omega_{d,6}(\cG)=3-\tfrac{N(\cG)}2-V_{4}+\rho(\cG).
\end{align}
The list of potentially divergent graphs is thus given by the
following table:
\renewcommand{\arraystretch}{1.5}
\begin{table}[!htb]
  \begin{displaymath}
  \begin{array}{c|cccc||cccccc}
    \multicolumn{1}{c}{}&\multicolumn{4}{c}{\vp^{4}_{6}}&\multicolumn{6}{c}{\vp^{6}_{5}}\\
    \hline
    N&2&\phantom{-}2&\phantom{-}2&4&2&\phantom{-}2&\phantom{-}2&4&\phantom{-}4&6\\
    \rho&0&-1&-2&0&0&-1&-2&0&-1&0\\
    \omega_{d}&2&\phantom{-}1&\phantom{-}0&0&2&\phantom{-}1&\phantom{-}0&1&\phantom{-}0&0
  \end{array}
\end{displaymath}
  \caption{Potentially divergent graphs}
  \label{tab:PotDivGraphs}
\end{table}
\renewcommand{\arraystretch}{1.0}

In the next section, we will characterize fully melonic graphs ($\rho(\cG)=0$) and
explain how to deal with the non fully melonic ones ($\rho(\cG)<0$).

\subsection{Classification of divergent graphs}
\label{sec:class-diverg-graphs}

We now describe the graphs of \cref{tab:PotDivGraphs} such that
$\rho=0,-1,-2$. To this aim, we first re-express the divergence degree
as follows.\\

\noindent
Let $\cG$ be an uncolored graph and $\cG_c$ be its colored
extension. We define
$\ot(\cG)\defi \sum_{J\subset\cG_{c}}g_{\widetilde{J}}$, where $\widetilde{J}$ is the pinched jacket associated with a jacket $J$ of $\cG_{c}$.
\begin{prop}[Divergence degree]
  \label{thm-Omega}
  The degree of divergence $\omega_{d}$ of a $P(\vpb,\vp)_{U(1)^{d}}$
  model with propagator \eqref{eq:prop} is given by
  \begin{align}
    \omega_{d}(\cG)=&(-2L+F-R)(\cG)\label{eq:Omega}\\
    =&-\frac{2}{(d-1)!}\big(\ot(\cG)-\omega(\partial\cG)\big)-(C_{\partial\cG}-1)-\frac{d-3}{2}N+(d-1)\nonumber\\
    &\qquad+\frac{d-3}{2}n\cdotaction
    V-(d-1)\cdotaction V-R\label{eq:Omegad}
  \end{align}
where $C_{\partial\cG}$ is the number of vertex-connected components of $\partial\cG$.
\end{prop}
\begin{proof}
The number of vertices $V(\cG_c)$ of the colored extension $\cG_c$ of
$\cG$ can be given in terms of $L(\cG)$ and $N(\cG)$ by the relation
$V(\cG_c)=n\cdotaction V=2L+N$. The number of its lines is
$L(\cG_c)=L+L_{i,\cG_c}\defi\frac{1}{2}[(d+1)n\cdotaction V-N],$ where
$L_{i,\cG_c}$ is the number of  internal lines of $\cG_c$ which do not
appear in $\cG$.  In the same  way $F(\cG_c)=F+F_{i,\cG_c}$. There
exist $d!/2$ jackets of $\cG_c$. Each face is shared by $(d-1)!$
jackets. Then $\sum_{J}F_J=(d-1)!F(\cG_c).$ The numbers of vertices
(resp. lines, resp. external edges) of $\cG_{c}$, $J$ and $\widetilde
J$ are equal. The graph $\widetilde{J}$ is a vacuum ribbon graph and its parameters $F_{\widetilde{J}}$, $V_{\widetilde{J}}$ and $L_{\widetilde{J}}$ satisfy the following relation
\begin{align}\label{quanty}
F_{\widetilde{J}}=F_{i,\widetilde{J}}+F_{e,\widetilde{J}}=2-2g_{\widetilde{J}}-V(\cG_c)+ L(\cG_c),\quad L(\cG_c)=L_{\widetilde{J}},\,\,\, V(\cG_c)=V_{\widetilde{J}},
\end{align}
where $F_{i,\widetilde{J}}$ is the number of internal faces of
$\widetilde{J}$, and $F_{e,\widetilde{J}}$ is the number of faces of
$\widetilde{J}$ which are made of external faces of $J$. Denote by
$F_{i,\widetilde{J},\cG}$ the  number of internal faces of
$\widetilde{J}$ colored $0i$, $1\les i\les d$ and
$F_{i,\widetilde{J},\cG_c}$ the number of internal faces colored $ij$,
$1\les i,j\les d$. We get  $F_{i,\widetilde{J}}=F_{i,\widetilde{J},\cG}+F_{i,\widetilde{J},\cG_c}.$ Then 
\begin{align}
\sum_{J} F_{i,\widetilde{J}}
=(d-1)!(F+F_{i,\cG_c}).
\end{align}
The number $F_{i,\cG_c}$ can be easily computed \citep{Gurau2012aa}
\begin{align}
F_{i,\cG_c}=\Big[\frac{(d-1)(d-2)}{2}\frac{n}{2}+d-1\Big]\cdotaction V.
\end{align}
The quantity $\sum_J (-V_{\widetilde{J}}+ L_{\widetilde{J}})$ can be written as using  \eqref{quanty}
\begin{align}\label{sum-jacket}
\sum_J -V_{\widetilde{J}}+ L_{\widetilde{J}}
=\frac{n\cdotaction V}{4}d!(d-1)-\frac{d!}{4}N(\cG).
\end{align}
Then
\begin{align}
F
=-\frac{1}{(d-1)!}\sum_J F_{e,\widetilde{J}}   -
\frac{2}{(d-1)!}\sum_J g_{\widetilde{J}}-\frac{(d-1)}{4}(4-2n)\cdotaction V-\frac{d}{4}N    +d .
\end{align}
The next stage consists in re-expressing $\sum_J F_{e,\widetilde{J}}$ in terms
of the parameters of the boundary graph $\partial \cG$ of $\cG$. For
any jacket $J_{\partial}$ of $\partial\cG$, note that
$V_{\partial \cG}=V_{J_\partial}=N,\,\,\, L_{\partial \cG}=L_{J_\partial }=F_{e},\,\,\, dV_{\partial \cG}=2L_{\partial \cG}\Rightarrow F_{e}=\frac{d}{2}N.$ There exist $(d-1)!/2$ boundary jackets of $\cG_c.$  Each face of the graph $\partial \cG$ is shared by exactly $(d-2)!$ boundary jackets. 
Using the fact that the Euler characteristic $\chi(J_{\partial})=2C_{J_\partial}-2g_{J_\partial}=
V_{J_\partial}-L_{J_\partial}+F_{J_\partial}$, we arrive at
\begin{align}\label{DineFabien}
F_{\partial \cG}=\frac{2}{(d-2)!}\sum_{J_\partial}C_{J_\partial}-\frac{2}{(d-2)!}\sum_{J_\partial}g_{J_\partial}+\frac{(d-1)}{2}\frac{(d-2)}{2}N.
\end{align}
Noting that $C_{J_\partial}=C_{\partial \cG}.$ Finally
\begin{align}
\sum_{J} F_{e,\widetilde{J}}&=(d-2)!F_{\partial \cG}\nonumber\\
&=(d-1)!(C_{\partial \cG}-1)-2\sum_{J_\partial}g_{J_\partial}+\frac{(d-1)!(d-2)}{4}N+(d-1)!
\end{align}
and
\begin{align}
  F=& -\frac{2}{(d-1)!}\Big(\sum_J
  g_{\widetilde{J}}-\sum_{J_\partial}g_{J_\partial}\Big)-(C_{\partial
    \cG}-1)\nonumber\\
  &\hspace{4cm}-\frac{d-1}{2}N+d-1-\frac{d-1}{4}(4-2n)\cdotaction V.\label{eq:Faces}
\end{align}
Using $L=\frac 12(n\cdotaction V-N)$ and equation \eqref{eq:Faces}, we get \eqref{eq:Omegad}.
\end{proof}

According to \cref{eq:Faces}, the number of internal faces of a graph
is given by
\begin{align}
  \label{eq:F}
  F(\cG)=&-\frac{2}{(d-1)!}(\ot(\cG)-\omega(\partial\cG))-(C_{\partial\cG}-1)+\frac{d-1}{2}\big(2-N+(n-2)\cdotaction
  V\big).
  \end{align}
We define
\begin{align}
\Fm(\cG)\defi&\frac{d-1}{2}\big(2-N+(n-2)\cdotaction V\big)=(d-1)(L-V+1)\\
\intertext{such that}
\Fm(\cG)-F(\cG)=&\frac{2}{(d-1)!}(\ot(\cG)-\omega(\partial\cG))+(C_{\partial\cG}-1).\label{eq:Fm-F}
\end{align}
According to lemma $5$ of \citep{Ben-Geloun2011aa} (or to
\cref{thm-ot-od} in \cref{sec:comb-analys-otcg}),
\begin{align}
  F(\cG)=\Fm(\cG)\Longleftrightarrow\ot(\cG)=\omega(\partial\cG)=C_{\partial\cG}-1=0.\label{eq:FFm}
\end{align}
Before giving the topological properties of the graphs with
$\rho=0,-1,-2$, we need the following definitions and technical lemma. Let us denote
the number of vacuum face-connected components of a graph $\cG$ by
$C^{f}_{0}(\cG)$. Let $\cG$ be a graph and $E$ a subset of its edges
equipped with a total order. We can thus write
$E=\{l_{1},\dots,l_{|E|}\}$. For all $i\in[|E|]^{*}\setminus\{1\}$, we
define $\cG_{i}\defi\cG/\{l_{1},\dotsc,l_{i-1}\}$ and $\cG_{1}\defi\cG$.
\begin{lemma}[Non-foaming $0$-dipoles]
  \label{thm-NonFoaming0Dip}
  Let $\cG$ be a vertex-connected non-vacuum ($N(\cG)>0$) uncolored
  $d$-tensor graph and $\cT$ any of its spanning trees. If there exists an order on the $\Lt$ lines of
  $\cR\defi\cG/\cT$ such that:
  \begin{enumerate}
  \item there exists $i_{0}\in[\Lt]^{*}$ such that $l_{i_{0}}$ is a
    $0$-dipole in $\cR_{i_{0}}$, and
  \item $C_{0}^{f}(\cR_{i_{0}+1})=C_{0}^{f}(\cR_{i_{0}})$,
  \end{enumerate}
  then $l_{i_{0}}$ is called a non-foaming $0$-dipole, and $\rho(\cG)\les -(d-2)$.
\end{lemma}
The proof requires another lemma proven in \citep{Carrozza2012ab}:
\begin{lemma}[Foaming $0$-dipoles]
  \label{thm-Foaming0Dip}
  Let $\cR$ be a rosette (i.e.\@ a one-vertex uncolored tensor graph)
  and $l$ a $0$-dipole in $\cR$. If $C_{0}^{f}(\cR/l)>C_{0}^{f}(\cR)$,
  then $\rho(\cR)=\rho(\cR/l)-(d-1)$.
\end{lemma}
\begin{proof}[of \cref{thm-NonFoaming0Dip}]
  Let us first suppose that the lemma is proven for face-connected
  graphs. Consider then a vertex-connected but
  face-disconnected graph $\cG$: $\cG=\bigcup_{i\in I}\cG^{i}$ and
  $\rho(\cG)=\sum_{i\in I}\rho(\cG^{i})$. At least one of the
  $\cG^{i}$'s contains a non-foaming $0$-dipole. The lemma is thus
  proven if all the other face-connected components satisfy $\rho\les
  0$. Fortunately, a vertex-connected but face-disconnected graph
  cannot have vacuum face-connected components. The color structure of
  the tensor graphs ensures it. And we conclude using \cref{thm-rho}.\\

  So let us assume that $\cG$ is face-connected and let us prove the lemma by induction on the number $\Lt$ of lines of
  $\cR$. If $\Lt=1$, $l_{1}$ is a $0$-dipole in $\cR$. In this case,
  $F(\cR)=F(\cG)=0=R(\cG)$ so that $\rho=-(d-2)$.

  Let us now assume that the lemma holds for all graphs with at most $\Lt=n$ lines and
  let us consider a graph with $\Lt=n+1$ edges. If $l_{1}$ is a
  $0$-dipole which does not create additional vacuum connected components, $\rho(\cR)=\rho(\cR/l_{1})-(d-2)$ (if
  $R(\cR/l_{1})=R(\cR)$) or $\rho(\cR)=\rho(\cR/l_{1})-(d-1)$ (if
  $R(\cR/l_{1})=R(\cR)-1$). Moreover
  $C_{0}^{f}(\cR/l_{1})=C_{0}^{f}(\cR)=C_{0}^{f}(\cG)=1$. Thus,
  acording to \cref{thm-rho}, $\rho(\cR/l_{1})\les 0$ and
  $\rho(\cR)=\rho(\cG)\les -(d-2)$.\\
  \noindent
  If $l_{1}$ is a $k$-dipole, $0\les k\les d-1$, which does not satisfy the conditions of the
  lemma, then $\rho(\cG)=\rho(\cR)=\rho(\cR/l_{1})-(d-k-1)$. The
  contraction of $l_{1}$ may have created $q$ connected components
  (i.e.\@ the number $C^{v}(\cR/l_{1})$ of vertex-connected components
  of $\cR/l_{1}$ is $q$) with $1\les q\les d-k$. But by assumption, at
  least one of these $q$ components obey the induction
  hypothesis. Then,
  \begin{align}
    \rho(\cG)\les&q-1-(d-2)-(d-k-1)\les -(d-2)
  \end{align}
which proves the lemma.
\end{proof}

We are now in position to give the topological properties of the
divergent graphs of the models \labelcref{eq:Action4}
and \labelcref{eq:Action6}.
\begin{prop}
  \label{thm-TopPropDivGraphs}
  The divergent graphs of the models \labelcref{eq:Action4}
  and \labelcref{eq:Action6} are classified in the following table
  \renewcommand{\arraystretch}{1.2}
  \begin{table}[!htp]
    \begin{displaymath}
  \begin{array}{c|c|c|c|c||c}
    &N&\ot(\cG)&\omega(\partial\cG)&C_{\partial\cG}-1&\omega_{d}(\cG)\\
    \hline
    \multirow{2}{*}{$\vp^{4}_{6}$} &2&0&0&0&2\\
    &4&0&0&0&0\\
    \hline
    \multirow{4}{*}{$\vp^{6}_{5}$}&2&0&0&0&2\\
    &4&0&0&0&1\\
    &4&0&0&1&0\\
    &6&0&0&0&0\\
    \hline
    \hline
  \end{array}
\end{displaymath}
    \caption{Classification of divergent graphs}
    \label{tab:ClassDivGraphs}
  \end{table}
\end{prop}
\begin{proof}
  Let $\cG$ be a graph of one of the types listed in
  \cref{tab:PotDivGraphs}. If $\rho(\cG)=0$, according to
  \cref{thm-rho}, $\cG$ is fully melonic and by
  \cref{thm-FullyMelonicDivDegree} and \cref{eq:FFm},
  $\ot(\cG)=\omega(\partial\cG)=C_{\partial\cG}-1$. Let us now assume
  that $\rho(\cG)<0$. If $R(\cG)<\Rm(\cG)$, according to \cref{thm-Foaming0Dip},
    for any tree $\cT$ in $\cG$ and any order on the lines of
    $\cG/\cT$, there must be a non-foaming $0$-dipole in $\cG$ and by
    \cref{thm-NonFoaming0Dip}, $\rho(\cG)\les -(d-2)\les -3$ for both
    models \labelcref{eq:Action4} and \labelcref{eq:Action6}.\\
    We can thus assume that $R(\cG)=\Rm(\cG)$. In this case (see \cref{eq:Fm-F}), 
    \begin{align}
      \rho(\cG)=F(\cG)-(d-1)\Lt(\cG)=-\frac{2}{(d-1)!}(\ot(\cG)-\omega(\partial\cG))-(C_{\partial\cG}-1).
    \end{align}
    But \citeauthor*{Ben-Geloun2012ab} have proven that for any $d$-tensor
    graph $\cG$, the quantity
    $\frac{2}{(d-1)!}(\ot(\cG)-\omega(\partial\cG))$ is either equal
    to zero or bigger or equal to $d-2$ \citep{Ben-Geloun2012ab}. Thus for $d\ges 5$, graphs $\cG$
    such that $\rho\ges -2$ and $R=\Rm$ must satisfy
    $\ot(\cG)=\omega(\partial\cG)=0$. Consequently, graphs with
    $\rho=-1$ (resp. $-2$) have a boundary graph with two
    (resp. three) (vertex-)connected components. We simply conclude
    the proof by noting that the boundary graph of a $2$-point graph is necessarily connected.
\end{proof}

\section{Renormalization}
\label{sec:renormalization}

Let us consider an arbitrary divergent graph $\cG$ with $N$ external legs.
 This graph has $N $  external propagators.
 We denote by $p_{f_e},\,$ the  external momentum of $\cG$ associated to the external  face $f_e$,
 and $P_{j}=(p_{j,f_{e_1}},p_{j,f_{e_2}},\cdots, p_{j,f_{e_d}}), \,\, 1\les j\les N$
 the $d$-vectors associated to the external edges of
 $\cG$. In the same manner, the $d$-dimensional momentum of an
 internal line $l$ of $\cG$ will be denoted by a capital letter:
 $P_{l}=(p_{f_{1}(l)},\dotsc,p_{f_{d}(l)})$.\\

\noindent In this section, we will complete the proof of the finiteness, order by
order, of the usual \emph{effective} series which express any
connected function of the theory in terms of an infinite set of
effective couplings, related one to each other by a discretized flow
\citep{Riv1}. Reexpressing these effective series in terms of the
renormalized couplings would reintroduce in the usual way the
Zimmermann’s forests of ``useless'' counterterms and build the standard renormalized series. The most explicit way to check
finiteness of these renormalized series in order to complete the ``BPHZ
theorem'' is to use the standard ``classification of forests'' which
distributes Zimmermann’s forests into packets such that the sum over
assignments in each packet is finite \citep{Riv1}. This part is
completely standard and will not be repeated here. As a consequence,
we can focus our attention on (primitively divergent) dangerous graphs (see \cref{sec:multiscale-analysis}).\\

The truncated amplitude of a graph $\cG$ with a scale attribution $\mu$ is given by
\begin{align}
\bar{A}^\mu_{\cG}=&\sum_{\{P_{j}\}}
\vpb_{P_{1}}\vp_{P_{2}}\cdots \vpb_{P_{N-1}}\vp_{P_{N}}A^\mu_{\cG}(\{P_j\}),
\end{align}
where
\begin{align}
A^\mu_{\cG}(\{P_j\})=&\sum_{P_l,
  l\in\cL}\int\prod_{l\in \mathcal{L}}\big( e^{-\alpha_l(a P_l^2+m^{2})}\delta_l(\sum_{j=1}^dp_{l,j})\big)\prod_{v\in\cV}K_{v}(\{P_{l}\}) \prod_{l\in \mathcal{L}}d\alpha_l
\end{align}
and the $\vp$'s and $\vpb$'s are fields of scales strictly lower than
the lowest internal scale of $\cG^{\mu}$.\\

Each delta function $\delta_l(\sum_{j=1}^d p_{l,j})$ can be re-expressed in the form
\begin{align}
\delta_l(\sum_{j=1}^d p_{l,j})=\delta_l(\sum_{f\in\mathcal{F}}\,\epsilon_{lf}\,p_{f}+\sum_{f_{e}\in\mathcal{F}_e}\,\widetilde\epsilon_{lf_{e}}\,p_{f_e}),\label{eq:deltaldeltaf}
\end{align}
where the tensor $\widetilde\epsilon_{lf_{e}}$ is the
tensor analogous to $\epsilon_{lf}$ but associated with the external
faces of $\cG$. Remark also that 
\begin{align}
\prod_{l\in\mathcal{L}}e^{-\alpha_l aP_l^2}=\prod_{f\in\mathcal{F}}
e^{-a(\sum_{l\in f }\alpha_l)p_f^2}\prod_{f_{e}\in\mathcal{F}_e}
e^{-a(\sum_{l\in f_{e} }\alpha_l)p_{f_e}^2}.\label{eq:PlPf}
\end{align}
In the rest of this work, we set $\alpha_f\defi \sum_{l\in f }\alpha_l$.

\subsection{Resolution of the delta functions}

Let $l$ be an arbitrary internal line of $\cG$ such that  $l\in
\mathcal{L}_\mu$, see \cref{sec:multiscale-analysis}. Recall that the
subset $\mathcal{L}_\mu$ of $\mathcal{L}$ is defined such that
$|\mathcal{L}_\mu|=\rk{\epsilon(\cG)}=R$. The number of delta
functions, such that the one in \cref{eq:deltaldeltaf}, that will interest us
is exactly the rank $R\les L$ of the matrix $(\epsilon)_{lf}.$ The
remaining of the delta functions, i.e. the $L-R$ delta functions, will
be  put to $1$ i.e. $\delta_l=\delta(0)=1$ after summation.\\

The kernels $K_{v}$ are such that the momenta are conserved along the strands:
\begin{align}
  A^{\mu}_{\cG}(\{P_{j}\})=&K_{\partial\cG}(\{P_{j}\})\sum_{p_{f},f\in\cF}\int\prod_{l\in
    \mathcal{L}}\big( e^{-\alpha_l(a
    P_l^2+m^{2})}\delta_l(\sum_{j=1}^dp_{l,j})\big)\prod_{l\in
    \mathcal{L}}d\alpha_l\\
  \fide&K_{\partial\cG}(\{P_{j}\})\cA^{\mu}_{\cG}(\{P_{j}\}).
\end{align}
The kernel $K_{\partial\cG}$ identifies the momenta at the two ends of
each of the $Nd/2$ external faces. Thus it precisely reproduces the
structure of the boundary graph $\partial\cG$ of $\cG$.\\

According to \cref{eq:deltaldeltaf,eq:PlPf} (see also \cref{sec:multiscale-analysis}),
\begin{align}
\sum_{p_{f},f\in\cF}\,\prod_{l\in \mathcal{L}}e^{-\alpha_la P_l^2}\delta_l(\sum_{j=1}^d\,\,p_{l,j})=&\sum_{p_{f},f\in\cF\setminus\cF_{\mu}}\prod_{f\in\mathcal{F}\setminus\mathcal{F}_\mu}
e^{-a\alpha_f p_f^2}\prod_{f_{e}\in\mathcal{F}_e}
e^{-a\alpha_{f_{e}} p_{f_e}^2}\nonumber\\
&\times \prod_{f\in\mathcal{F}_\mu}e^{-a\alpha_f(\sum_{f'\in
    \mathcal{F},f'\neq f}\epsilon_{l(f)f'}p_{f'}+\sum_{f_{e}\in\cF_{e}}\widetilde\epsilon_{l(f)f_{e}}p_{f_{e}})^{2}}.
\end{align}
Finally
\begin{align}
\bar{A}^\mu_{\cG}=&\sum_{P_{j},j\in[N]^{*}}
K_{\partial\cG}(\{P_{j}\})\,\vpb_{P_{1}}\vp_{P_{2}}\cdots \vpb_{P_{N-1}}\vp_{P_{N}}\prod_{f_{e}\in\mathcal{F}_e}
e^{-a\alpha_{f_{e}} p_{f_e}^2}\int\prod_{l\in \mathcal{L}}d\alpha_{l}\,e^{-a\alpha_l m^2}\nonumber\\
&\times\sum_{p_{f},f\in\cF\setminus\cF_{\mu}}\prod_{f\in\mathcal{F}\setminus \mathcal{F}_\mu}
e^{-a\alpha_fp_f^2}\prod_{f\in\mathcal{F}_\mu}e^{-a\alpha_f(\sum_{f'\in
    \mathcal{F},f'\neq f}\epsilon_{l(f)f'}p_{f'}+\sum_{f_{e}\in\cF_{e}}\widetilde\epsilon_{l(f)f_{e}}p_{f_{e}})^{2}}.
\end{align}

\subsection{Taylor Expansions}\label{sec:taylor-expansions}
The aim of this section is to expose general features of the Taylor
expansion of the Feynman amplitudes.\\

\noindent Let $\cG$ be any Feynman graph of the models \labelcref{eq:Action4}
and \labelcref{eq:Action6}. $\cG$ may not have a divergent amplitude. We define the parametrized amplitude $\cA^\mu_{\cG}(\{P_{j}\},t)$ which depends on a parameter $t\in[0,1]$ such that
$\cA^\mu_{\cG}(\{P_{j}\},t)\defi \cA^\mu_{\cG}(\{tP_{j}\})$. Obviously,
$\cA^\mu_{\cG}(\{P_{j}\})=\cA^\mu_{\cG}(\{P_{j}\},1)$. We will perform a
Taylor expansion (in $t$) of $\cA^\mu_{\cG}(\{P_{j}\},1)$ around $t=0$.

\subsubsection{Zeroth order}
\begin{align}
  \cA^\mu_{\cG}(\{P_{j}\})=&\cA^\mu_{\cG}(\{P_{j}\},t)|_{t=1}\\
  =&\prod_{f_{e}\in\mathcal{F}_e}
e^{-a\alpha_{f_{e}} t^{2}p_{f_e}^2}\int\prod_{l\in \mathcal{L}}d\alpha_{l}\,e^{-a\alpha_l m^2}\sum_{p_{f},f\in\cF\setminus\cF_{\mu}}\prod_{f\in\mathcal{F}\setminus \mathcal{F}_\mu}
e^{-a\alpha_fp_f^2}\nonumber\\
&\left.\times\prod_{f\in\mathcal{F}_\mu}e^{-a\alpha_f(\sum_{f'\in
    \mathcal{F},f'\neq f}\epsilon_{l(f)f'}p_{f'}+\sum_{f_{e}\in\cF_{e}}\widetilde\epsilon_{l(f)f_{e}}tp_{f_{e}})^{2}}\rabs_{t=1}.
\end{align}
The zeroth order term of the Taylor expansion of $\cA^\mu_{\cG}(\{P_{j}\})$ is
\begin{align}
\cA^\mu_{\cG,0}(\{P_{j}\})\defi\int\prod_{l\in \mathcal{L}}d\alpha_{l}\,e^{-a\alpha_l m^2}\sum_{p_{f},f\in\cF\setminus\cF_{\mu}}\prod_{f\in\mathcal{F}\setminus \mathcal{F}_\mu}
e^{-a\alpha_fp_f^2}\prod_{f\in\mathcal{F}_\mu}e^{-a\alpha_f(\sum_{f'\in
    \mathcal{F},f'\neq f}\epsilon_{l(f)f'}p_{f'})^{2}}.
\end{align}
Note that it is independant of the $P_{j}$'s. The Taylor expansion of
$\cA_{\cG}^{\mu}$ induces an expansion of $\bar A_{\cG}^{\mu}$ whose
zeroth order takes the following form:
\begin{align}
  \bar A^{\mu}_{\cG,0}\defi \cA^\mu_{\cG,0}\sum_{\{P_{j}\}}K_{\partial\cG}(\{P_{j}\})
\vpb_{P_{1}}\vp_{P_{2}}\cdots \vpb_{P_{N-1}}\vp_{P_{N}}.\label{eq:A0}
\end{align}
In conclusion, the zeroth order term of $\bar A_{\cG}^{\mu}$ has the
form of a vertex whose connecting pattern is given by the boundary
graph of $\cG$.

\subsubsection{First order}
\label{sec:first-order}

The first order of the Taylor expansion of $\cA_{\cG}^{\mu}$ is
\begin{align}
  \cA^{\mu}_{\cG,1}(\{P_{j}\})\defi&\left.\frac{d\cA^{\mu}_{\cG}(\{P_{j},\cdot\})}{dt}\rabs_{t=0}.\label{eq:cA1def}
\end{align}
To simplify notations, let us introduce, for all $f\in\cF_{\mu}$
\begin{align}
  \kp_{f}\defi&\sum_{f'\in\mathcal{F},f'\neq
    f}\epsilon_{l(f)f'}p_{f'}\text{ and } \kp_{e(f)}\defi\sum_{f_{e}\in\cF_{e}}\widetilde\epsilon_{l(f)f_{e}}p_{f_{e}}.\label{eq:kPfkPe}
\end{align}
Thus we get
\begin{equation}
  \begin{split}
    \cA^{\mu}_{\cG,1}(\{P_{j}\})=-2a \int\prod_{l\in
      \mathcal{L}}d\alpha_{l}\,e^{-a\alpha_l
      m^2}\sum_{p_{f},f\in\cF\setminus\cF_{\mu}}\prod_{f\in\mathcal{F}\setminus
      \mathcal{F}_\mu}
    e^{-a\alpha_fp_f^2}\\
    \times \prod_{f\in\cF_{\mu}}e^{-a\alpha_f\kp_f^2}\big(\sum_{f\in\cF_{\mu}}\alpha_{f}\kp_{e(f)}\kp_{f}\big).
  \end{split}\label{eq:cA1}
\end{equation}
The sums on the $p_{f}$'s are performed over $\Z$ and the summands are
odd so that $\cA^{\mu}_{\cG,1}$ vanishes identically.

\subsection{Traciality of the counterterms}
\label{sec:trac-count}

In \citep{Bonzom2012ac}, it has been realized that the effective
action for a single tensor field, obtained by the integration of $d$
tensor fields out of the $d+1$ fields of an iid model, is dominated
by invariant traces indexed by melonic $d$-colored graphs. The
vertices of the model \eqref{eq:Action4} (resp.\@ \eqref{eq:Action6})
correspond to all the vacuum \emph{connected} melonic $6$-colored
(resp.\@ $5$-colored) graphs upto order $4$ (resp. $6$) plus a
so-called anomaly namely a product of two quadratic traces.\\

\noindent{} We consider the divergent graphs of the $\vp^{4}_{6}$ and
$\vp^{6}_{5}$ models, listed in \cref{tab:ClassDivGraphs}. For
simplicity, let us start with the graphs $\cG$ such that
$\omega_{d}(\cG)=0$ or $1$. Those graphs have $4$ or $6$ external
legs. According to the discussion of \cref{sec:taylor-expansions},
$\bar A_{\cG,0}^{\mu}$ corresponds to a vertex whose structure is
given by the boundary graph of $\cG$. All the divergent graphs in our
models have melonic boundary graphs. If $N(\cG)=4$ and
$C_{\partial\cG}=1$, $\partial\cG$ is one of the graphs depicted in
\cref{fig:Vertices4}. If $N(\cG)=6$ and $C_{\partial\cG}=1$,
$\partial\cG$ is one of the graphs of \cref{fig:Vertices6}. Finally,
there are $4$-point divergent graphs with a disconnected melonic
boundary. They correspond to the disconnected invariant trace of
\cref{fig:VerticeV42}. Such an ``anomaly'' has also been observed in
\citep{Ben-Geloun2011aa}.

As $\cA^{\mu}_{\cG,1}=0$, we have
\begin{subequations}
  \begin{align}
    \cA^{\mu}_{\cG}(\{P_{j}\})=&\cA^{\mu}_{\cG,0}+\int_{0}^{1}(1-s)\left.\frac{d^{2}\cA^{\mu}_{\cG}(\{P_{j},\cdot\})}{dt^{2}}\rabs_{t=s}ds\fide
      \cA^{\mu}_{\cG,0}+\cR_{2},\label{eq:TaylorOmegad01}
    \end{align}
    \begin{equation}
      \begin{split}
        \cR_{2}=\int_{0}^{1}(1-s)\int\prod_{l\in
          \mathcal{L}}d\alpha_{l}\,e^{-a\alpha_l
          m^2}\sum_{p_{f},f\in\cF\setminus\cF_{\mu}}\prod_{f\in\mathcal{F}\setminus
          \mathcal{F}_\mu}
        e^{-a\alpha_fp_f^2}\prod_{f_{e}\in\mathcal{F}_e}
        e^{-a\alpha_{f_{e}} s^{2}p_{f_e}^2}\\
        \times\prod_{f\in\cF_{\mu}}e^{-a\alpha_f(\kp_f+s\kp_{e(f)})^2}\lsb\Big(\sum_{f_{e}\in\cF_{e}}-2a\alpha_{f_{e}}sp_{f_{e}}^{2}+\sum_{f\in\cF_{\mu}}-2a\alpha_{f}\kp_{e(f)}(\kp_{f}+s\kp_{e(f)})\Big)^{2}\right.\\
      \left.+\sum_{f_{e}\in\cF_{e}}-2a\alpha_{f_{e}}p_{f_{e}}^{2}+\sum_{f\in\cF_{\mu}}-2a\alpha_{f}\kp_{e(f)}^{2}\rsb.
      \end{split}\label{eq:R2}
    \end{equation}
  \end{subequations}
$\cR_{2}$ is the renormalized amplitude of $\cG^{\mu}$. Let us prove
that it is finite (in fact summable with respect to its scale index). Using the simple upper bound
\begin{equation}
  |p_{f}|e^{-a\alpha_{f}p_{f}^{2}}\les\frac{e^{-a\alpha_{f}p_{f}^{2}/2}}{\sqrt{a\alpha_{f}}},\label{eq:pexpBound}
\end{equation}
one easily gets that the terms between square bracket in \cref{eq:R2}
are bounded by $cM^{-2(i_{\cG}(\mu)-e_{\cG}(\mu))}$ where $c$ is a
positive constant. The rest of the summand/integrand reproduces the
power counting of $\cG$ (see \cref{sec:multiscale-analysis}). Thus for
logarithmically or linearly divergent graphs, $\cR_{2}$ is finite.\\

Let us now consider the divergent $2$-point graphs of the
models \labelcref{eq:Action4} and \labelcref{eq:Action6}. Their degree
of divergence $\omega_{d}$ equals $2$. In consequence, their amplitude has to be
expanded upto order $2$:
\begin{subequations}
  \begin{align}
    \cA^{\mu}_{\cG}(\{P_{j}\})=&\cA^{\mu}_{\cG,0}+\cA^{\mu}_{\cG,2}(\{P_{j}\})+\cR_{3},\label{eq:2ptExpansion}\\
    \cR_{3}=&\tfrac 12
    \int_{0}^{1}(1-s)^{2}\left.\frac{d^{3}\cA^{\mu}_{\cG}(\{P_{j},\cdot\})}{dt^{3}}\rabs_{t=s}ds.
    \end{align}
  \end{subequations}
Let us recall that (see \cref{eq:R2})
\begin{subequations}
  \begin{align}
    \frac{d^{2}\cA^{\mu}_{\cG}(\{P_{j},t\})}{dt^{2}}=&\int\prod_{l\in
      \mathcal{L}}d\alpha_{l}\,e^{-a\alpha_l
      m^2}\sum_{p_{f},f\in\cF\setminus\cF_{\mu}}\prod_{f\in\mathcal{F}\setminus
      \mathcal{F}_\mu}
    e^{-a\alpha_fp_f^2}\nonumber\\
    &\qquad\times\prod_{f_{e}\in\mathcal{F}_e} e^{-a\alpha_{f_{e}}
      t^{2}p_{f_e}^2}\prod_{f\in\cF_{\mu}}e^{-a\alpha_f(\kp_f+t\kp_{e(f)})^2}[E(t)^{2}+E'],\\
    E(t)\defi&\sum_{f_{e}\in\cF_{e}}-2a\alpha_{f_{e}}tp_{f_{e}}^{2}+\sum_{f\in\cF_{\mu}}-2a\alpha_{f}\kp_{e(f)}(\kp_{f}+t\kp_{e(f)}),\label{eq:Et}\\
    E'\defi&\sum_{f_{e}\in\cF_{e}}-2a\alpha_{f_{e}}p_{f_{e}}^{2}+\sum_{f\in\cF_{\mu}}-2a\alpha_{f}\kp_{e(f)}^{2}.\label{eq:Eprime}
  \end{align}
\end{subequations}
Note that $E'$ does not depend on $t$. As a consequence,
\begin{align}
  \frac{d^{3}\cA^{\mu}_{\cG}(\{P_{j},t\})}{dt^{3}}=&\int\prod_{l\in
      \mathcal{L}}d\alpha_{l}\,e^{-a\alpha_l
      m^2}\sum_{p_{f},f\in\cF\setminus\cF_{\mu}}\prod_{f\in\mathcal{F}\setminus
      \mathcal{F}_\mu}
    e^{-a\alpha_fp_f^2}\nonumber\\
    &\qquad\times\prod_{f_{e}\in\mathcal{F}_e} e^{-a\alpha_{f_{e}}
      t^{2}p_{f_e}^2}\prod_{f\in\cF_{\mu}}e^{-a\alpha_f(\kp_f+t\kp_{e(f)})^2}\big(E[E^{2}+E']+2EE'\big).
\end{align}
We have already seen that $|E(t)|\sim
M^{-(i_{\cG}(\mu)-e_{\cG}(\mu))}$ and $|E'|\sim
M^{-2(i_{\cG}(\mu)-e_{\cG}(\mu))}$. Thus $|\cR_{3}|$ is bounded by
$M^{-3(i_{\cG}(\mu)-e_{\cG}(\mu))}$ times the power counting of
$\cG^{\mu}$ and is therefore summable for $i_{\cG}(\mu)>e_{\cG}(\mu)$.\\

$\bar A^{\mu}_{\cG,0}$ has the structure of the boundary graph of
$\cG$. As $N(\cG)=2$, its boundary is the unique melon with two
vertices and $\bar A^{\mu}_{\cG,0}$ thus contributes to the
renormalization of the mass.\\

There only remains to prove that $\bar A^{\mu}_{\cG,2}$ renormalizes
the wave function. The argument is a bit subtle and twofold.
\begin{subequations}
  \begin{align}
    \cA^{\mu}_{\cG,2}=&\int\prod_{l\in
      \mathcal{L}}d\alpha_{l}\,e^{-a\alpha_l
      m^2}\sum_{p_{f},f\in\cF\setminus\cF_{\mu}}\prod_{f\in\mathcal{F}\setminus
      \mathcal{F}_\mu}
    e^{-a\alpha_fp_f^2}\prod_{f\in\cF_{\mu}}e^{-a\alpha_f\kp_f^2}\,[E(0)^{2}+E']\label{eq:AG2}\\
    \fide&\sum_{f_{1},f_{2}\in\cF_{\mu}}\kp_{e(f_{1})}\kp_{e(f_{2})}F_{1}(f_{1},f_{2})+\sum_{f_{e}\in\cF_{e}}p_{f_{e}}^{2}F_{2}(f_{e})+\sum_{f\in\cF_{\mu}}\kp_{e(f)}^{2}F_{3}(f),\label{eq:F1F2F3}\\
    F_{1}(f_{1},f_{2})=&4a^{2}\int\prod_{l\in
      \mathcal{L}}d\alpha_{l}\,e^{-a\alpha_l
      m^2}\alpha_{f_{1}}\alpha_{f_{2}}\nonumber\\
    &\hspace{3.5cm}\times\sum_{\substack{p_{f}\\f\in\cF\setminus\cF_{\mu}}}\kp_{f_{1}}\kp_{f_{2}}\prod_{f\in\mathcal{F}\setminus
      \mathcal{F}_\mu}
    e^{-a\alpha_fp_f^2}\prod_{f\in\cF_{\mu}}e^{-a\alpha_f\kp_f^2},\label{eq:F1}\\
    F_{2}(f_{e})=&-2a \int\prod_{l\in
      \mathcal{L}}d\alpha_{l}\,e^{-a\alpha_l
      m^2}\alpha_{f_{e}}\sum_{p_{f},f\in\cF\setminus\cF_{\mu}}\prod_{f\in\mathcal{F}\setminus
      \mathcal{F}_\mu}
    e^{-a\alpha_fp_f^2}\prod_{f\in\cF_{\mu}}e^{-a\alpha_f\kp_f^2},\label{eq:F2}\\
    F_{3}(f)=&-2a\int\prod_{l\in
      \mathcal{L}}d\alpha_{l}\,e^{-a\alpha_l
      m^2}\alpha_{f}\sum_{p_{f},f\in\cF\setminus\cF_{\mu}}\prod_{f\in\mathcal{F}\setminus
      \mathcal{F}_\mu}
    e^{-a\alpha_fp_f^2}\prod_{f\in\cF_{\mu}}e^{-a\alpha_f\kp_f^2}.\label{eq:F3}
  \end{align}
\end{subequations}
$\cA^{\mu}_{\cG,2}$ contributes to the renormalization of the wave
function if it is of the form
$F\sum_{f_{e}\in\cF_{e}}p_{f_{e}}^{2}$ where $F$ is a constant independant of
the $f_{e}$'s. We will see in the sequel that is not but that the
models are still renormalizable. We will need to exploit the fully
melonic character of the $2$-point divergent graphs and a
non-perturbative argument.\\

First of all, let us remark that none of the $F_{i}$'s are
constant. Moreover the first and third terms in \cref{eq:F1F2F3} do
not seem to be sums of squares of $p_{f_{e}}$\!'s. Let us first study
the third term. According to its definition,
\cref{eq:kPfkPe}, $\kp_{e(f)}$ is \emph{in general} a sum of external
momenta. Let us prove that in the case of fully melonic graphs, this sum contains at most one term. Indeed,
according to the definition of the sets $\cF_{\mu}$ and $\cL_{\mu}$
(see \cref{sec:multiscale-analysis}), to any internal face
$f\in\cF_{\mu}$, we associate a unique internal line
$l(f)\in\cL_{\mu}$ such that $l(f)\in f$. According to the definition
of the matrix $\widetilde{\epsilon}$, $\kp_{e(f)}$ is the (possibly
alternating) sum of momenta of the external faces to which the line
$l(f)$ contributes. So we have to prove that a line in $\cL_{\mu}$
contributes to at most one external face.\\
As proven in \cref{thm-TreeContraction}, for any spanning tree $\cT$
in $\cG$, the rows of $\epsilon$
corresponding to tree lines are linear combinations of the loop
lines. In other words, $\cL_{\mu}\subset\cL(\cG)\setminus\cL(\cT)$
(remember that $\cL_{\mu}$ is a set of maximally independant
edges). Let us then contract a spanning tree and consider the rosette
$\cG/\cT$. This contraction does not change the nature (internal or
external) of the faces to which the lines of $\cL_{\mu}$
contribute. As $\cG$ is fully melonic, there exists an order on the
edges of $\cG/\cT$ such that for all $i\in[L(\cG/\cT)]$, $l_{i}$ is a
$(d-1)$-dipole in $\cG_{i}$, see \cpageref{thm-NonFoaming0Dip}. Thus
each $l_{i}$ contributes to $d-1$ internal faces (of length $1$) and
to possibly one external face.\\

In consequence, for any internal face $f\in\cF_{\mu}$, there exists at
most one external face $f_{e}(f)$ such that
$\kp_{e(f)}=\widetilde{\epsilon}_{l(f)f_{e}(f)}p_{f_{e}(f)}$. The
third and first term of \cref{eq:F1F2F3} rewrites
\begin{subequations}
  \label{eq:F1F3Rewriting}
  \begin{align}
    \sum_{f\in\cF_{\mu}}\kp_{e(f)}^{2}F_{3}(f)=&\sum_{f\in\cF_{\mu}}p^{2}_{f_{e}(f)}F_{3}(f)=\sum_{f_{e}\in\cF_{e}}p^{2}_{f_{e}}\sum_{\substack{f\in\cF_{\mu},\\f_{e}(f)=f_{e}}}F_{3}(f),\label{eq:F3Rewriting}\\
    \sum_{f_{1},f_{2}\in\cF_{\mu}}\kp_{e(f_{1})}\kp_{e(f_{2})}F_{1}(f_{1},f_{2})=&\sum_{f_{e,1},f_{e,2}\in\cF_{e}}p_{f_{e,1}}p_{f_{e,2}}\sum_{\substack{f_{1},f_{2}\in\cF_{\mu},\\f_{e}(f_{1})=f_{e,1}\\f_{e}(f_{2})=f_{e,2}}}F_{1}(f_{1},f_{2}).\label{eq:F1Rewriting}
  \end{align}
\end{subequations}
Unfortunately, the term with $F'_{1}$ still does not seem to be a sum of squares of
external momenta. In fact it is and it is once more due to the fact
that $\cG$ is fully melonic. Let us prove the following simple result:
\begin{lemma}
  \label{thm-IntFacesWithDifftExtFaces}
  Let $\cG$ be a fully melonic $d$-tensor graph. Let $f_{e,1}$ and
  $f_{e,2}$ be two (not necessarily different) external faces of $\cG$. Let
  $l$ (resp.\@ $l'$) be a loop line contributing to
  $f_{e,1}$ (resp.\@ $f_{e,2}$). Then,
  \begin{equation}
    \{f\in\cF\tqs l\in f\}\cap \{f\in\cF\tqs l'\in f\}=\varnothing.\label{eq:EmptyIntersetion}
  \end{equation}
\end{lemma}
In words, if, in a fully melonic graph, there are two loop lines
contributing to two external faces, then they contribute to
no common internal face.
\begin{proof}
  It goes by induction on the lines of $\cG/\cT$. There exists an
  order on $\cL(\cG/\cT)$ such that for all $i\in[L(\cG/\cT)]$,
  $l_{i}$ is a $(d-1)$-dipole in $\cG_{i}$. Without loss of
  generality, let us assume that $l=l_{i}$ and $l'=l_{j}$ with
  $i<j$. In $\cG_{i}$, $l_{i}$ is a $(d-1)$-dipole. Then all the
  internal faces to which $l_{i}$ contributes are of length $1$ in
  $\cG_{i}$. In particular $l_{j}$ does contribute to no internal face
  of $l_{i}$.
\end{proof}
Let us now consider \cref{eq:F1Rewriting}. Let $f_{e,1},f_{e,2}$ be
two \emph{different} external faces of $\cG$. Let $f_{1},f_{2}$ be two internal faces
of $\cG$ such that $f_{e}(f_{i})=f_{e,i}$ for $i=1,2$. Then
$l(f_{1})\neq l(f_{2})$ and these lines do not share any
internal face. As a consequence, the sums in $\kp_{f_{1}}$ and in
$\kp_{f_{2}}$ have no term in common. The summand in $F_{1}(f_{1},f_{2})$ is thus
odd under the simultaneous change of sign of all the momenta in
$\kp_{f_{1}}$ (say) and $F_{1}(f_{1},f_{2})=0$ in this case.\\

\Cref{eq:F3Rewriting} rewrites
\begin{equation}
  \sum_{f_{1},f_{2}\in\cF_{\mu}}\kp_{e(f_{1})}\kp_{e(f_{2})}F_{1}(f_{1},f_{2})=\sum_{f_{e}\in\cF_{e}}p_{f_{e}}^{2}\sum_{\substack{f_{1},f_{2}\in\cF_{\mu},\\f_{e}(f_{1})=f_{e}(f_{2})=f_{e}}}F_{1}(f_{1},f_{2}).\label{eq:F3Rewriting2}
\end{equation}
All three terms in \cref{eq:F1F2F3} have now been proven to be sums of
squares of external momenta. But the coefficients of these quadratic
polynomials still depend on the external faces. And this not an
artefact. These sums contain only external faces wich are made of
internal lines. In other words, external faces of length $0$ do not
appear. And there are, of course, graphs with external faces of length
$0$ (see \cref{UncolEx1} for an example). $\cA^{\mu}_{\cG,2}$ cannot
in general reproduce a $p^{2}$ term.\\
Fortunately, the interactions we have considered are symmetric under
any permutation of the colors $1$ to $d$ (the \emph{positive} colors). The
external faces of a $2$-point graph are indexed by the colors from $1$ to
$d$: $\{f_{e}\in\cF_{e}\}=\{f_{e,01},f_{e,02},\dots,f_{e,0d}\}$. Moreover
the set of permutations on $[d]$ (or the set of a given tye of
interaction) can be partitionned into the equivalence classes under the
action of the cyclic permutations. We say that two graphs are
equivalent if the colored extension of one of them can be obtained from
the colored extension of the other by a cyclic permutation of the
positive colors. Let $[\cG]$ be the set of
representatives of such an equivalence class (thus $\cG,\cG'\in[\cG]$
are such that $\cG_{c}$ can be obtained from $\cG'_{c}$ by a cyclic
permutation of the positive colors).\\

\noindent According to the discussion above, $\cA^{\mu}_{\cG,2}$ is of the form
\begin{equation}
  \cA^{\mu}_{\cG,2}=\sum_{f_{e}\in\cF_{e}}p_{f_{e}}^{2}F_{\cG}(f_{e}).
\end{equation}
Thus,
\begin{equation}
  \sum_{\cG'\in[\cG]}\cA^{\mu}_{\cG',2}=p^{2}\sum_{f_{e}\in\cF_{e}}F_{\cG}(f_{e}).
\end{equation}
The second order of the taylor expansion of the sum of the amplitudes
of all the graphs in $[\cG]$ contribute to the wave function
renormalization which finally concludes the proof of the perturbative
renormalizability of the models \labelcref{eq:Action4} and \labelcref{eq:Action6}.

\newpage
\section{The super-renormalizable \texorpdfstring{$\vp^4_5$-}{ }model}
\label{sec:super-ren} 

The analysis of the divergence degree in \cref{sec:analys-diverg-degr}
provides us with another model of potential interest 
that we now describe. Let us consider the  $\vp^4_5$ tensor model with the same dynamics
described so far and quartic interaction as given by \eqref{eq:vertex2}.
This model can be viewed as well as a truncation of the $\vp^6_5$ to a smaller
set of interactions.

 Using equation \eqref{eq:FullymelonicDegree}, we can deduce that the
 divergence degree of a fully melonic graph is
\begin{align}
\omega_d(\cG)=-(N-6)-2V.
\end{align}
\begin{prop}
The rank-5 $\vp^4_5$ tensor model is super-renormalizable.
\end{prop}
\begin{proof}
If the quantity $\widetilde{\omega}(\cG)-\omega(\partial\cG)>0$
i.e.\@ not all jackets of $\cG_c$ are planar,  then 
\begin{align}
\omega_d(\cG)\les 2-(C_{\partial \cG}-1)-N-2V_2-R.
\end{align}
Using the fact that $V_2\ges 0$, $C_{\partial \cG}\ges 1$ and $R\ges 1$, we get $\omega_d(\cG)\les 1-N.$ This shows that  non melonic graphs are convergent graphs.  In contrast, if the quantity $\widetilde{\omega}(\cG)-\omega(\partial\cG)=0$
then  $C_{\partial \cG}=1$  and we get $\omega_d(\cG)\les 4-N-R.$ The divergent graphs have exactly  two external legs. Therefore $\omega_d(\cG)=2-V.$ The divergent graphs of this model are given in \cref{fig:uncolGraphphi5}.
So we infer that the $\vp^4_5$ tensor model is super-renormalizable like the $\vp^4_4$ model studied in \cite{Carrozza2012aa}.
\end{proof}
\begin{figure}[!htp]
  \centering
 \centering
  \subfloat[  $\omega_d=1$]{{\label{UGphi51}}\includegraphics[scale=.6]{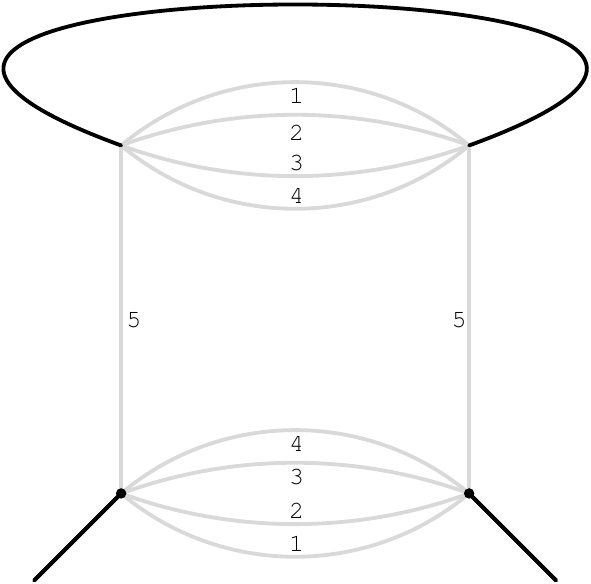}}\hspace{4cm}
  \subfloat[  $\omega_d=0$]{{\label{UGphi52}}\includegraphics[scale=.6]{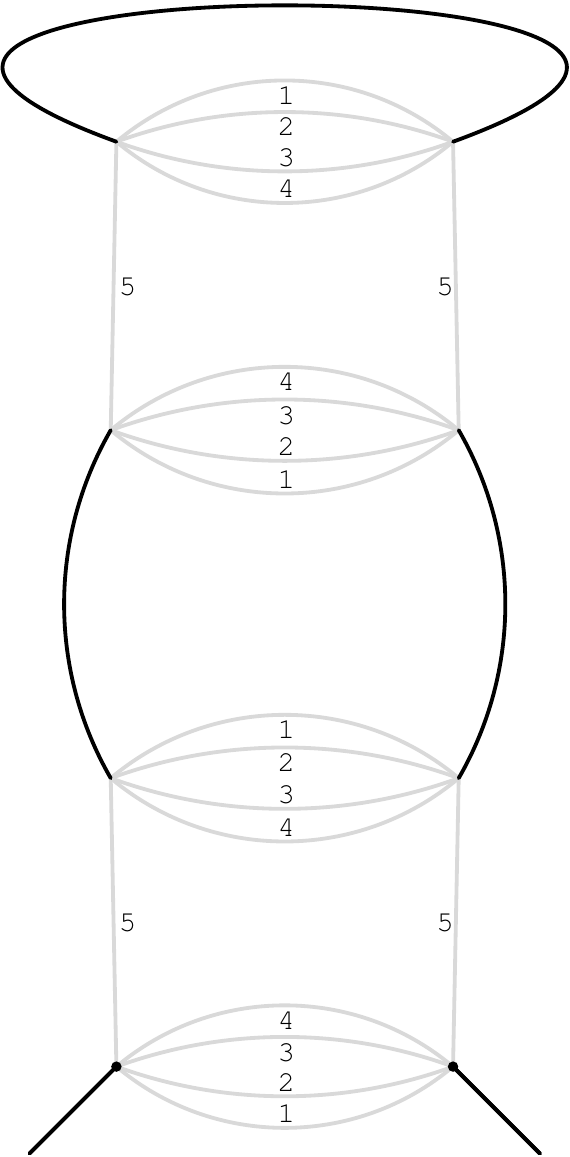}}
\caption{Divergent graphs of $\vp^4_5$}
  \label{fig:uncolGraphphi5}
\end{figure}

\section{Conclusion and discussion}

Just renormalizability is a property shared by all physical interactions except (until now) gravity. In the renormalization group sense it is natural. Indeed just renormalizable interactions survive long-lived renormalization group flow. They can be considered the result of a kind of Darwinian selection associated to such flows. Therefore if quantum gravity can be renormalized, it will rely on the same powerful technique that applies successfully to all other interactions of the standard model \cite{Rivasseau2012ab}.

In this work, we have shown that the $\vp_6^4$ and $\vp^6_5$ tensor models are renormalizable  at all orders of perturbation. The central point of this proof is given by the multiscale analysis. Our result sheds more light on 
the power counting in TGFTs with the gauge invariance condition.
This gauge condition had already been introduced in the previous work of Carrozza et al \cite{Carrozza2012aa} who  showed that
the generic rank-four models are super-renormalizable. The hurdle 
which can appear in the power counting due to the emergence of connected components in the $k$-dipole contraction is fully resolved now.  This work and previous results \cite{Ben-Geloun2011aa, Carrozza2012aa,  Carrozza2012ab} shows that there is indeed a neat family of renormalizable TGFT.

Having defined the first just renormalizable tensor models satisfying 
the gauge invariance, it remains to address the interesting question 
about how from such renormalizable models, one can recover General Relativity in the continuum limit. A phase transition from discrete to continuum geometries, from discrete degrees of freedom in the form of basic simplex 
(dual to tensors) presented here to more elaborate ones, should be understood.
This phase transition would be a conceivable scenario if, for instance,
the models described here can be proved asymptotically free
in the UV such that the renormalized coupling constants become
larger and larger in the opposite direction. 
Some tensor models without gauge invariance have been proved to 
be asymptotically free 
\cite{Ben-Geloun2012aa,Geloun2012ab,Geloun2012aa}. 
The study of the $\beta$-functions
of the $\vp^4_6$ and $\vp^6_5$ characterizing the UV limit 
of these models will be addressed in forthcoming works.

\paragraph{Acknowledgements}

\label{sec:acknowledgements}

The authors are indebted to Vincent Rivasseau for having proposed us the
problem treated here and for his guidance through the
stranded meanders of TGFT. They also sincerely thank Joseph Ben Geloun
for his numerous, complete and rapid explanations of his work. D. Ousmane Samary thank the Centre
international de math\'ematiques pures et appliqu\'ees (CIMPA), the Labex Milyon, the Association pour la Promotion Scientifique de l'Afrique (APSA) and the Laboratoire de Physique Th\'eorique d'Orsay (LPT) for financial supports.

\newpage
\appendix
\section{Paths in a graph}
\label{sec:paths-graph}

This section aims at illustrating the different definitions introduced
for the proof of \cref{thm-TreeContraction}. We choose a graph and
depicts its vertices as  black dots, see \cref{fig:OrientedGraphEx}.
\begin{figure}[!htp]
  \centering
  \includegraphics[scale=0.8]{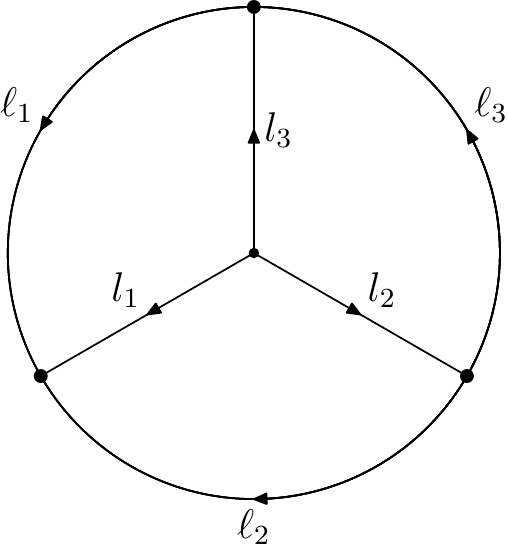}
  \caption{An oriented graph}
  \label{fig:OrientedGraphEx}
\end{figure}

Let us consider the oriented face $f=(\ell_{1},\ell_{2},\ell_{3})$. We
have:
\begin{align}
  \veps(f)=\bordermatrix{
    &l_{1}&l_{2}&l_{3}\cr
    \ell_{1}&\phantom{-}1&\phantom{-}0&-1\cr
    \ell_{2}&-1&\phantom{-}1&\phantom{-}0\cr
    \ell_{3}&\phantom{-}0&-1&\phantom{-}1
  },\quad
\eta=
      \bordermatrix{
        &l_{1}&l_{2}&l_{3}\cr
        l_{1}&\phantom{-}1&\phantom{-}0&\phantom{-}0\cr
        l_{2}&\phantom{-}0&\phantom{-}1&\phantom{-}0\cr
        l_{3}&\phantom{-}0&\phantom{-}0&\phantom{-}1\cr
        \ell_{1}&\phantom{-}1&\phantom{-}0&-1\cr
        \ell_{2}&\phantom{-}1&-1&\phantom{-}0\cr
        \ell_{3}&\phantom{-}0&-1&\phantom{-}1
      },\quad
      \epsilon=\bordermatrix{
        &f\cr
        l_{1}&\phantom{-}0\cr
        l_{2}&\phantom{-}0\cr
        l_{3}&\phantom{-}0\cr
        \ell_{1}&\phantom{-}1\cr
        \ell_{2}&-1\cr
        \ell_{3}&\phantom{-}1
      }.
    \label{eq:EpsilonsAndEta}
\end{align}
Note that we have three paths denoted by $\mathcal{P}_{\cT}(\ell_1)=\{l_{3-}, l_{1+}\}$,
$\mathcal{P}_{\cT}(\ell_2)=\{l_{2-},l_{1+}\}$
$\mathcal{P}_{\cT}(\ell_3)=\{l_{2-},l_{3+}\}.$ The signs $+$ and $-$ are used to identify the direction on the path $\mathcal{P}_{\cT}(\ell_i),\,i=1,2,3$ of the path-lines  $l_{i}$ with respect to the direction of $\ell_i.$ This is well illustrated in the first formula of equation \eqref{eq:EpsilonsAndEta}. If $\veps(f)_{l\ell}=0$ then $l\notin\mathcal{P}_{\cT}(\ell)$.  One sees easily that 
\begin{align}
\epsilon_{\ell_{1}f}\eta_{\ell_{1}l_1}+
\epsilon_{\ell_{2}f}\eta_{\ell_{2}l_1}+
\epsilon_{\ell_{3}f}\eta_{\ell_{3}l_1}=0\\
\epsilon_{\ell_{1}f}\eta_{\ell_{1}l_2}+
\epsilon_{\ell_{2}f}\eta_{\ell_{2}l_2}+
\epsilon_{\ell_{3}f}\eta_{\ell_{3}l_2}=0\\
\epsilon_{\ell_{1}f}\eta_{\ell_{1}l_3}+
\epsilon_{\ell_{2}f}\eta_{\ell_{2}l_3}+
\epsilon_{\ell_{3}f}\eta_{\ell_{3}l_3}=0.
\end{align}
Therefore $\sum_{\ell}\epsilon_{\ell f}\eta_{\ell l}=0 $ and then relation $\epsilon_{lf}=-\sum_{\ell\in\cL,\,\ell\neq l}\eta_{\ell
    l}\epsilon_{\ell f}$ is well satisfied.
\section{Combinatorial analysis of
  \texorpdfstring{$\ot(\cG)-\omega(\partial\cG)$}{Omega - Boundary Omega}}
\label{sec:comb-analys-otcg}

We propose here an alternative purely combinatorial proof of the fact
that  $\ot(\cG)-\omega(\partial\cG)\ges 0$. This proof is simpler than
the analysis of \citep{Ben-Geloun2012ab}. However it only proves a
weaker bound when $\ot(\cG)-\omega(\partial\cG)>0$ and $d>4.$ In the
case where $d=4$ the bounds of  \citep{Ben-Geloun2012ab} and this
appendix ($(d-1)!$) happen to coincide. The sign of
$\ot(\cG)-\omega(\partial\cG)$ can be analyzed using the so-called
dipole contraction. We immediately remind the reader with the
definition of a $0k$-dipole \citep{Ben-Geloun2011aa}.
\begin{defn}[$0k$-dipole]
  A $0k$-dipole (where $k=0,1,\cdots,d-1$) of a colored graph
  $\cG_{c}$ is a set of $k+1$ lines, one of which of color $0$, joining
  the same two vertices and such that no other lines connect the same
  two vertices.
\end{defn}
The contraction of a $0k$-dipole erases the $k+1$ lines of the dipole and connects the remaining $d-k$
lines on both sides of the dipole by respecting the colors. See 
\cref{fig:uncolGraphC}. Let us denote by $\cG_c'$ the graph obtained after
contraction of a $0k$-dipole of $\cG_c$. We have
\begin{align}\label{Contraction}
V(\cG_c')=V(\cG_c)-2,\,\,\,\,\, L(\cG_c')=L(\cG_c)-(d+1).
\end{align}
Let us consider a $0k$-dipole inside the colored graph $\cG_c$.
A ``pair'' is a couple of colors $(i,j)$, $i,j=0,1,2,\dots,d$. 
If none of the $k+1$ lines of the dipole bears color $i$ or $j$, the
pair is said to be ``outer''. If exactly one of the lines of the
dipole bears color $i$ or $j$, the pair is ``mixed''. If one line of
the dipole has color $i$ and another one color $j$, the pair is ``inner''.
\begin{figure}[!htp]
  \centering
  \subfloat[$0k$-dipole]{{\label{UGphiC1}}\raisebox{-.5\height}{\includegraphics[scale=1.2]{UncolGraphs-13}}}\hspace{2cm}
  \subfloat[After contraction]{{\label{UGphiC2}}\raisebox{-.5\height}{\includegraphics[scale=1.2]{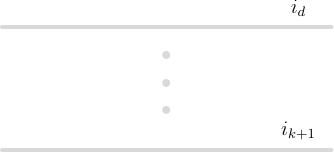}}}
\caption{Contraction of a $0k$-dipole}
  \label{fig:uncolGraphC}
\end{figure}
An outer pair $(i,j)$ is said to be of type A or disconnected by
the dipole contraction if the half-edges of lines $i$ and $j$ at each
corner on the left- and on the right-hand side of the dipole belong to
two different connected components of the graph after the dipole
contraction. Outer pairs belonging necessarily to closed faces, they
are single-faced in $\cG_{c}$. The pair $(i,j)$ is said to be special
if the half-edges of lines $i$ and $j$ belong to one single connected
component of $\cG'_{c}$. There are two types of special pairs. Type B
outer pairs are single-faced in $\cG_{c}$ (hence double-faced in
$\cG'_{c}$). Type C outer pairs are double-faced in $\cG_{c}$
(single-faced in $\cG'_{c}$).

After contraction, $F(\cG_c)$ has increased by 1 for each pair of type
$A$ or $B$ and decreased by $1$ for pairs of type $C$. Remark that the
mixed pairs preserve the number of faces. In the same manner the
number of faces decreases by $1$ for each internal pair. We then
arrive at
\begin{align}
F(\cG'_c)-F(\cG_c)=A+B-C-I
\end{align}
where $X\in\{A,B,C\}$ is the number of faces of type $X$ and $I$ is
the number of inner faces.\\

The strategy is the same as the one in \citep{Ben-Geloun2011aa}. We
will bound the difference between $\ot(\cG)$ and
$\ot(\cG')$. Then we apply the same bound all along a sequence of
dipole contractions from $\cG$ to $\partial\cG$ (remember that the graph
obtained after contraction of all the dipoles of $\cG$ is essentially
$\partial\cG$ \citep{Ben-Geloun2011aa}). 
\begin{prop}[Bound on genera]
  \label{thm-genera} Let $J$ be a jacket of $\cG_{c}$. We note $J'$
  the jacket of $\cG'_{c}$ corresponding to the same permutation as
  $J$. With $c'$ the number of connected components of $\cG'_{c}$,
\begin{align} 
\sum_{J}(g_{\widetilde{J}}-g_{\widetilde{J'}})\ges\frac{(d-1)!}{2}(d-k-c')(c'+k-1)\ges 0.\label{eq:ot-otprimePositive}
\end{align}
\end{prop}
\begin{proof}
Using  relations \eqref{Contraction}, the Euler characteristics of $\widetilde{J}$ and $\widetilde{J}'$ are given by
\begin{align}
2-2g_{\widetilde{J}}=V-L+F_{\widetilde{J}},\quad 2c'-2g_{\widetilde{J'}}=V-2-(L-d-1)+F_{\widetilde{J'}}.
\end{align}
We get 
\begin{align}
\sum_{J}(g_{\widetilde{J}}-g_{\widetilde{J'}})=\frac{1}{2}
\Big[\sum_{J}(F_{\widetilde{J}'}-F_{\widetilde{J}})+\frac{d!(d-1)}{2}
-d!(c'-1)\Big].
\end{align}
Recall that $\sum_{J}F_{\widetilde{J}}=(d-1)!F(\cG_c)$ and  $\sum_{J'}F_{\widetilde{J'}}=(d-1)!F(\cG'_c).$ Then
\begin{align}
\sum_{J}(g_{\widetilde{J}}-g_{\widetilde{J'}})
&=\frac{(d-1)!}{2}\Big[F(\cG'_c)-F(\cG_c)+\frac{d(d-1)}{2}
-d(c'-1)\Big]\cr
&=\frac{(d-1)!}{2}\big(A+B-C-I\big)+\frac{d!(d-1)}{4}
-\frac{d!}{2}(c'-1).
\end{align}
The rest of the proof  will be devoted to find a lower bound on the
quantity $A+B-C-I.$ This can be done using the formalism of
integer partitions. The number of connected components $c'$ of
$\cG'_{c}$ being fixed, the $d-k$ external lines of the dipole are
distributed among $c'$ connected colored graphs \citep{Ben-Geloun2011aa}.
Each such configuration corresponds to a partition of $d-k$
into $c'$ parts. Let $\cP_{p}(n)$ be the set of partitions of $n$ in $p$
parts:
\begin{align}
  \label{eq:PartitionEnsemble}
  \cP_{p}(n)\defi&\big\{(n_{i})_{1\les i\les p},\,n_{1}\ges
  n_{2}\ges\dotsm\ges n_{p}\tqs\sum_{i=1}^{p}n_{i}=n\big\}.
\end{align}
For all $n\in\N^{*}$ and $1\les p\les n$, we denote by $\lambda_{1}$ the
following partition of $\cP_{p}(n)$:
\begin{align}
  \label{eq:Lambda1}
  \lambda_{1}\defi&(n-p+1,\underbrace{1,\dots,1}_{p-1\text{ terms}}).
\end{align}
Given a configuration of the external lines of a $0k$-dipole, that is to
say a partition $\lambda=(n_{i})$ of $d-k$ into $c'$ parts, we
have
\begin{subequations}
  \begin{align}
    (B+C)(\lambda)=&\sum_{i=1}^{c'}\frac{n_{i}(n_{i}-1)}{2}=\tfrac12\sum_{i=1}^{c'}n_{i}^{2}-\tfrac12(d-k),\label{eq:BC(Lambda)}\\
    A(\lambda)=&\tfrac12(d-k)(d-k-1)-(B+C),\label{eq:A(Lambda)}\\
    I=&\tfrac12(k+1)k\label{eq:I(k)}.
  \end{align}
\end{subequations}
For a $0k$-dipole and a fixed $c'$, $\ot(\cG)-\omega(\partial\cG)$ is
minimal when $B+C$ is maximal and $B=0$. Therefore,
\begin{align}
  \begin{split}
    \ot(\cG)-\ot(\cG')\ges&\frac{(d-1)!}{2}\Big(\tfrac12(d-k)(d-k-1)-2C_{M}-\tfrac12(k+1)k\Big)\\
    &+\frac{d!(d-1)}{4}-\frac{d!}{2}(c'-1),\label{eq:ot-oMin}
  \end{split}\\
  C_{M}\defi&\max_{\lambda\in\cP_{c'}(d-k)}(B+C)(\lambda).
\end{align}
It remains to determine $C_{M}$. To this aim, we note that
\begin{prop}
  \label{thm-PartitionDescent}
  Any partition of $\cP_{p}(n)$ can be obtained from $\lambda_{1}$ by a
(possibly empty) sequence of the following basic operation $D_{ij}$: let
$\lambda=(n_{i})\in\cP_{p}(n)$. If there exists a couple
$(i,j)\in([p]^{*})^{2},\,i<j$ such that $n_{i}-n_{j}\ges 2$, we define
$D_{ij}\lambda=\lambda^{(1)}=(n^{(1)}_{i})\in\cP_{p}(n)$ by
$n^{(1)}_{i}=n_{i}-1$, $n^{(1)}_{j}=n_{j}+1$, and for all $k\neq i,j$,
$n^{(1)}_{k}=n_{k}$. We potentially need to reorder the
$n^{(1)}_{k}$'s to get a proper partition.
\end{prop}
\begin{proof}
  Let us consider a partition $\lambda=(n_{i})\in\cP_{p}(n)$. If
  $\lambda=\lambda_{1}$, we are done. If not, it is enough to
  prove that there exists $\lambda^{(-1)}\in\cP_{p}(n)$ and $(i,j)\in([p]^{*})^{2}$ such that
  $D_{ij}\lambda^{(-1)}=\lambda$. We get the proposition simply by
  iterating that result.

  The construction of $\lambda^{(-1)}$ goes as follows. As
  $\lambda\neq\lambda_{1}$, there is
  $(i,j)\in([p]^{*})^{2},\,i<j$ such that $n_{i},n_{j}\ges 2$.
  $\lambda^{(-1)}=(n'_{k})$ is then defined as: $n'_{i}=n_{i}+1$,
  $n'_{j}=n_{j}-1$, for all $k\neq i,j$, $n'_{k}=n_{k}$. As
  $n'_{i}-n'_{j}\ges 2$, $D_{ij}\lambda^{(-1)}=\lambda$.
\end{proof}
If $\cP_{p}(n)$ is equipped with the lexicographical (total) order,
$\lambda_{1}$ is the highest partition. Moreover for all $\lambda$,
$D_{ij}\lambda<\lambda$. But
$\sum_{k=1}^{p}\big(n_{k}^{2}-(n^{(1)}_{k})^{2}\big)=2(n_{i}-n_{j}-2)\ges
0$. Thus the maximum over $\cP_{p}(n)$ of
$\sum_{i=1}^{p}n_{i}^{2}(\lambda)$ is reached for the highest partition in
the lexicographical order, namely $\lambda_{1}$. As a consequence,
\begin{align}
  C_{M}=&\tfrac12\big((d-k-c'+1)^{2}+c'-1\big)-\tfrac12(d-k)\label{eq:CM}\\
  \intertext{and}
  \begin{split}
    \ot(\cG)-\ot(\cG')\ges&\frac{(d-1)!}{2}\Big(\tfrac12(d-k)(d-k-1)-(d-k-c'+1)(d-k-c')\\
    &-\tfrac12k(k+1)\Big)+\frac{d!(d-1)}{4}-\frac{d!}{2}(c'-1)
  \end{split}\\
  =&\frac{(d-1)!}{2}(d-k-c')(k+c'-1).
\end{align}
As $1\les c'\les d-k$, $\ot(\cG)-\ot(\cG')\ges 0$.
\end{proof}
\begin{cor}
  \label{thm-ot-od}
  For any graph $\cG$, $\ot(\cG)-d\,\omega(\partial\cG)\ges 0$.
\end{cor}
\begin{proof}
  Let us denote $\cG/\cL(\cG)$ the graph obtained after a complete
  sequence of contractions of the dipoles of $\cG$. Iterating the
  bound \eqref{eq:ot-otprimePositive}, we get
  \begin{align}
    \ot(\cG)-\ot(\cG/\cL)\ges 0.
  \end{align}
The colored extension $(\cG/\cL(\cG))_{c}$ of $\cG/\cL$ is $\partial\cG_{c}$ equipped with
  external legs of color $0$. Let $\sigma=(\sigma(1)\dotsm\sigma(d))$ be a
  cyclic permutation on $[d]^{*}$ and $J_{\partial}(\sigma)$ the
  corresponding jacket of $\partial\cG_{c}$. Any permutation $\tau$ on $[d]$ of the
  following set (of cardinality $d$):
  \begin{align}
    P_{\sigma}\defi&\lb(0\sigma(1)\dots\sigma(d)),(\sigma(1)0\sigma(2)\dotsm\sigma(d)),\dots,(\sigma(1)\dotsm\sigma(d-1)0\sigma(d))\rb\label{eq:sigma0}
  \end{align}
  gives rise to a jacket $J(\tau)$ of $\cG_{c}$ such that
  $g_{\widetilde J(\tau)}=g_{J_{\partial}(\sigma)}$. Moreover the set
  of cyclic permutations on $[d]$ can be partitioned as
  $\cup_{\sigma\text{ on }[d]^{*}}P_{\sigma}$. Thus,
  \begin{align}
    \ot(\cG/\cL)=\sum_{J\subset(\cG/\cL)_{c}}g_{\widetilde J}=d\sum_{J_{\partial}\subset\partial\cG_{c}}g_{J_{\partial}},
  \end{align}
which ends the proof.
  \end{proof}

{\footnotesize
\bibliographystyle{fababbrvnat}
\bibliography{biblio-articles,biblio-books}
}
\bigskip\bigskip
\contactfabetal[D. O. S., F. V.-T.]

\contact[D. O. S.]{International Chair in Mathematical Physics and
  Applications (ICMPA-UNESCO Chair), 072 BP 50 Cotonou, Republic of
  Benin.}{dine.ousmanesamary@cipma.uac.bj}

\end{document}